\documentclass[aps,twocolumn,superscriptaddress,10pt]{revtex4-1}
\pdfoutput=1
\PassOptionsToPackage{usenames,dvipsnames}{color}
\usepackage{mathtools}
\usepackage{import}
\usepackage{hhline}
\usepackage{hyperref}
\usepackage{enumitem}
\usepackage{braket}
\usepackage{amsfonts}
 \usepackage{amsmath, amsthm}
\usepackage{color}
 \usepackage{float}
 \usepackage{times}

 \newtheorem{theorem}{Theorem}
 \newtheorem{proposition}[theorem]{Proposition}

 \newtheorem{lemma}[theorem]{Lemma}
 \newtheorem{corollary}[theorem]{Corollary}
 \newtheorem{definition}[theorem]{Definition}

 \newtheorem{remark}[theorem]{Remark}

\newcommand{\mc}[1]{\mathcal{#1}}
\newcommand{\mr}[1]{\mathrm{#1}}
\newcommand{\mb}[1]{\mathbb{#1}}

\newcommand{\tr}{\mathrm{Tr}} 
\newcommand{\Trp}[2]{\mathrm{Tr}_{#1}\left[ #2 \right]}
\newcommand{\supp}{\mathrm{supp}}

\newcommand{\id}{\mb{I}}

\newcommand{\argmin}{\mathrm{argmin}}

\newcommand{\mcH}{\mc{H}}
\newcommand{\mcB}{\mc{B}}

\newcommand{\mcS}{\mc{S}}

\newcommand{\N}{\mb{N}}

\newcommand{\R}{\mb{R}}
\newcommand{\CC}{\mb{C}}

\newcommand{\norm}[1]{\left\Vert #1 \right\Vert}

 \newcommand{\proj}[1]{\ket{#1}\bra{#1}}
 
\newcommand{\Dmax}[1]{D_{\text{max}}^{#1}}

\newcommand{\BH}[1]{\mcB \left( {#1} \right)}
\newcommand{\SH}[1]{\mcS \left( {#1} \right)}
\newcommand{\chn}{\mc{E}}
\newcommand{\co}{CO$_2$}
\usepackage[normalem]{ulem}

\newcommand{\fu}{Dahlem Center for Complex Quantum Systems, Freie Universit{\"a}t Berlin, 14195 Berlin, Germany}
\newcommand{\imperial}{Department of Computing, Imperial College London, London SW7 2AZ, U.K.}
\newcommand{\ucl}{Department of Physics and Astronomy, University College London, London WC1E 6BT, U.K.}
\newcommand{\sing}{School of Physical and Mathematical Sciences, Nanyang Technological University, 637371 Nanyang, Singapore}

\begin{document}
\title{Bounding the resources for thermalizing many-body localized systems}
\author{Carlo Sparaciari} 
\affiliation{\imperial}
\affiliation{\ucl}
\author{Marcel Goihl}
\affiliation{\fu}
\author{Paul Boes}
\affiliation{\fu}
\author{Jens Eisert}
\email{jense@zedat.fu-berlin.de}
\affiliation{\fu}
\author{Nelly Huei Ying Ng}
\affiliation{\fu}
\affiliation{\sing}
\begin{abstract}
Understanding under which conditions physical systems thermalize is a long-standing question in many-body physics. While generic quantum systems thermalize, there are known instances where thermalization is hindered, for example in many-body localized (MBL) systems. Here we introduce a class of stochastic collision models coupling a many-body system out of thermal equilibrium to an external heat bath. We derive upper and lower bounds on the size of the bath required to thermalize the system via such models, under certain assumptions on the Hamiltonian. We use these bounds, expressed in terms of the max-relative entropy, to characterize the robustness of MBL systems against externally-induced thermalization. Our bounds are derived within the framework of resource theories using the convex split lemma, a recent tool developed in quantum information. We apply our results to the disordered Heisenberg chain, and numerically study the robustness of its MBL phase in terms of the required bath size.
\end{abstract}
\maketitle
\section{Introduction}
When pushed out of equilibrium, closed interacting quantum many-body systems generically 
relax to an equilibrium state, where local subsystems can be described using thermal ensembles that only depend on the energy of the initial state. 
While this behaviour is plausible from the perspective of quantum statistical mechanics~\cite{Deutsch1991,
Srednicki1994,Popescu06,1408.5148,ngupta_Silva_Vengalattore_2011}, it is far from
clear which local properties are responsible for the emergence of thermalization. The discovery of
many-body localization (MBL) offers a fresh perspective into this question, as this effect occurs
in interacting many-body systems preventing them from actually thermalizing~\cite{Oganesyan,Schreiber}.
Examples of systems that are non-thermalizing, like integrable systems, have been known before, but these
have been fine-tuned such that small perturbations restore thermalization. Many-body localization is strikingly
different in this respect, as its non-thermalizing behaviour appears to be robust to changes in the
Hamiltonian~\cite{abanin_colloquium:_2019}. A related open question, recently considered in several papers,
is whether MBL is stable with respect to its own dynamics when small ergodic regions are present~\cite{luitz_how_2017,
ponte_pedro_thermal_2017,hetterich_noninteracting_2017,barisic_incoherent_2009,Goihl19}, or
when the system is in contact with an actual external environment~\cite{nandkishore_spectral_2014,
fischer_dynamics_2016,levi_robustness_2016,johri_many-body_2015}. This is both a critical question
for the experimental realization of systems exhibiting MBL properties, and for its fundamental implications
on the process of thermalization in quantum systems.
\par
The main focus of this work is to investigate the robustness of the MBL phase under instances of external dissipative processes. To do so, we introduce a physically-realistic class of interaction models, describing the interaction between a many-body system and a finite-sized thermal environment. Within these models, the interactions are described in terms of energy-preserving stochastic collisions occurring between the system (or regions thereof) and sub-regions of the bath. During the interactions, system and bath can be either weakly or strongly coupled. For this class of processes, we are able to derive analytical bounds on the minimum size of the bath required to thermalize the many-body system. We apply these bounds to the setting where the system is in the MBL phase, so as to characterize the robustness of this phase with respect to the coupling with an external environment. It is woth noting, however, that the bounds obtained hold for general many-body systems out of thermal equilibrium.
\par
It is key to the approach taken here -- and one of the merits of this work -- that in order to arrive at our results, 
we make use of tools from quantum information theory, tools that might at first seem somewhat alien to the 
problem at hand, but which turn out to provide a powerful machinery. We demonstrate this by using a 
technical result known as the convex split
lemma~\cite{anshu_quantum_2017,anshu_quantifying_2018}, to derive the quantitative bounds on the bath size required
for a region of the spin lattice to thermalize. Given that MBL phases are challenging to study theoretically,
and most known results are numerical in nature \cite{Luitz,Devakul,PhysRevLett.118.017201,Prosen,Pollmann,Wahl,AugstineMBL}, 
our work provides a fresh approach in understanding
such phases from an analytical perspective. The convex split lemma has originally been  derived in the context
of quantum Shannon theory, which is the study of compression and transmission rates of quantum information.
Our main contribution is to connect this mathematical result to a class of thermodynamic models
which can be used to describe thermalization processes in quantum systems. This gives rise to surprisingly
stringent and strong results. Note, however, in the approach taken it is assumed that systems thermalize close to exactly, a requirement that will be softened in future work.
\par
As part of our results, we find that the max-relative entropy~\cite{datta_min-_2009,tomamichel_quantum_2016}
and its smoothed version emerge as operationally significant measures, that quantify the robustness of the MBL phase
in a spin lattice. The max-relative entropy is an element within a family of entropic measures
that generalize the R{\'e}nyi divergences~\cite{erven_renyi_2014} to the quantum setting.
In order to illustrate the practical relevance of our results, we consider a specific system exhibiting
the MBL phase, namely the disordered Heisenberg chain. Employing exact diagonalization, we numerically
compute the value of the max-relative entropy as a function of the disorder and of
the size of the lattice region that we are interested in thermalizing. Our findings suggest that the MBL phase is
robust to thermalization despite being coupled to a finite external bath, 
under our collision models, indicating that such models allow for a conceptual understanding 
of the MBL phase stability.
This extends the narrative of Refs.~\cite{nandkishore_spectral_2014,fischer_dynamics_2016,levi_robustness_2016}, which
find that MBL is thermalized when coupled to an infinite sized bath.
Moreover, our numerical simulations show that the max-relative entropy signals the transition from the ergodic to 
the MBL phase. 
\section{Results}
\subsection{Thermalization setting}\label{sec:setting}
We first set up some basic notation. Given some Hamiltonian $H_S$ of a system $S$ with Hilbert space $\mcH_S$,
we define the thermal state with respect to inverse temperature $\beta = 1/(k_B T)$ as the quantum state
\begin{align}
    \tau_\beta(H_S) := \frac{e^{-\beta H_S}}{\tr(e^{-\beta H_S})}.
\end{align}
In what follows, we model the process of thermalization of $ S $ with an external heat bath $B$. In particular, let $H_S$
and $ H_B $ denote the Hamiltonian of $ S $ and $ B $ respectively. If $S$ is initially in a state $\rho$, then,
for fixed $\beta$ and any $\epsilon > 0$, we say that a global
process $\chn : \BH{\mcH_S \otimes \mcH_B} \rightarrow \BH{\mcH_S \otimes \mcH_B}$ $\epsilon$-thermalizes
the system $S$ if
\begin{equation}
\label{eq:def_thermalizing}
\norm{\chn \left( \rho \otimes \tau_\beta(H_B) \right) - \tau_\beta(H_S) \otimes \tau_\beta(H_B)}_1 \leq \epsilon,
\end{equation}
where $\norm{\cdot}_1$ is the trace norm. Intuitively, this corresponds to the situation where the process $\chn$ 
acts on the compound of the initial state of the system $\rho$ and the bath state $\tau_\beta(H_B)$, and brings
the system state close to its thermal state $\tau_\beta(H_S)$ while leaving the bath mostly invariant.
We write 
\begin{equation}
\rho\xrightarrow[H_B,\epsilon]{\mathcal{E}} \tau_\beta(H_S)
\end{equation} if Eq.~\eqref{eq:def_thermalizing} holds.
It is worth noting that $\epsilon$-thermalization requires the global system $SB$ to be close to thermal after
the channel $\chn$ is applied. Monitoring the bath as well is necessary in order to avoid the possibility of trivial
thermalization processes in which the non-thermal state $\rho$ is simply swapped into the bath,
which would merely move the excitation out of the considered region, rather than describing a physically realistic
dissipation process. Thus, our notion of thermalization is different from previously considered ones,
where for instance the sole system's evolution is considered. At the same time, and as mentioned before,
it is a rather stringent measure,
in that close to full global thermalization is required. 
\par
Having introduced the basic notation and terminology, we now turn to the model used in this work. We consider
a spin lattice $V$, where each site is described by a finite-dimensional Hilbert space $\mcH$. The
Hamiltonian of the system is composed of local operators, i.~e.,
\begin{equation}
\label{eq:lattice_ham}
H_V = \sum_x H_x, 
\end{equation}
where $x$ is labeling a specific subset of
adjacent sites in the lattice and $H_x$ is the corresponding Hamiltonian operator whose support is limited to these sites. 
Within the lattice, we consider a local region $R \subseteq V$ with Hilbert space $\mcH_R$. We are interested in the
stability of the MBL phase with respect to stochastic collisions between the lattice region $R$ and an external thermal
bath $B$. In order to precisely re-cast this problem in terms of $\epsilon$-thermalization, we first need to detail our
choices for the initial state of the region $R$, the Hamiltonians $H_B$ and $H_R$, and the class of maps $\chn$ that
model the thermalization process.
\par
Given an initial state vector of the lattice $\ket{\psi(0)}$, if we consider closed evolution generated by the Hamiltonian dynamics, then the global system is always in a pure state. However, if the local subsystems eventually equilibrate, then the equilibrium state is given by partially tracing over the global infinite-time average~\cite{gogolin_equilibration_2016},
\begin{equation}
\label{eq:inf_time_average}
\omega := \lim_{T \to \infty} \frac{1}{T} \int_0^T \mr{d} t \, \ket{\psi(t)}\bra{\psi(t)},
\end{equation}
so that the state that describes region $R$ at the time of its first interaction with the bath is
$\omega_R = \Trp{R^c}{\omega}$, where we trace out the remaining of the lattice $R^c = V \backslash R$. To define a valid Hamiltonian $H_R$, a natural approach is often to disregard interactions between $R$ and the rest of the lattice $R^c$. As such, we consider the Hamiltonian $H_R = \sum_{x: x \subseteq R} H_x$, which includes only terms $H_x$ whose support is contained in $R$, and denote the corresponding thermal state $\tau_\beta(H_R)$. It is worthwhile to point out that that there is an alternative natural approach to defining thermal states of subsystems, namely $\hat\tau_\beta (H_R) = \tr_{R^c}\left[\tau_\beta(H_V)\right]$ the reduced state over the complement $R$ of the thermal state of the full lattice~\cite{kliesch_locality_2014}. These two thermal states are close to each other whenever the interaction terms between $ R $ and $ R^c $ in $ H_V$ are small. During our later simulations, we check the values of max-relative entropy using both versions of thermal states, and find that they produced similar values, implying that these different alternatives are actually not too dissimilar from each other for the disordered Heisenberg chain, which is expected, given that interactions are 2-local.
\par
There has been a large body of work in thermalization, where a large class of many-body systems can effectively act as their own “environment”~\cite{1408.5148,gogolin_equilibration_2016}, thus local observables tend to equilibrate towards the corresponding thermal values. However, there are exceptions to this, in particular whenever there are non-negligible interactions between subsystems, of which MBL systems are an example. Nevertheless, most such systems still do equilibrate, and this is the assumption we make throughout the paper, that $\omega$ as defined in Eq.~\eqref{eq:inf_time_average} exists.

We model the external thermal bath as a collection of $n-1$ copies of the region $R$ in thermal equilibrium. More
formally, $B$ is a system with Hilbert space $\mcH_B = \mcH_R^{\otimes n-1}$ and Hamiltonian
\begin{equation}
\label{eq:H_B_copiesofregion}
H_B =
\sum_{i=1}^{n-1} H_R^{(i)},
\end{equation} 
where the operator $H_R^{(i)} = \id_1 \otimes \ldots \otimes \id_{i-1} \otimes H_R \otimes
\id_{i+1} \otimes \ldots \otimes \id_{n-1}$ only acts non-trivially on the $i$-th subsystem of the bath.
With this choice of Hamiltonian, the initial state on $B$ is $\tau_\beta(H_B) = \tau_{\beta}(H_R)^{\otimes n-1}$.  
Such a choice for the bath is crucial to make our problem analytically tractable, but is also physically relevant
for experimental setups, where it is possible to engineer one dimensional systems
which are then coupled to a bath per site \cite{bordia2016coupling,bordia2017probing}
or a mixture of a system and bath species that interact via contact interactions \cite{abadal2018bath}. Moreover, we note that  for the model of system-bath interactions that we introduce below, any state transition that can be realised on $R$ with any heat bath Hamiltonian $H_B$ can also be realized, for some $n$, with a Hamiltonian of the form \eqref{eq:H_B_copiesofregion}~\cite{horodecki2013fundamental}. We also refer the interested reader to Supplementary Note 3, for a further discussion on bath choices.
\begin{figure}[t]
\centering
\includegraphics[width=0.4\textwidth]{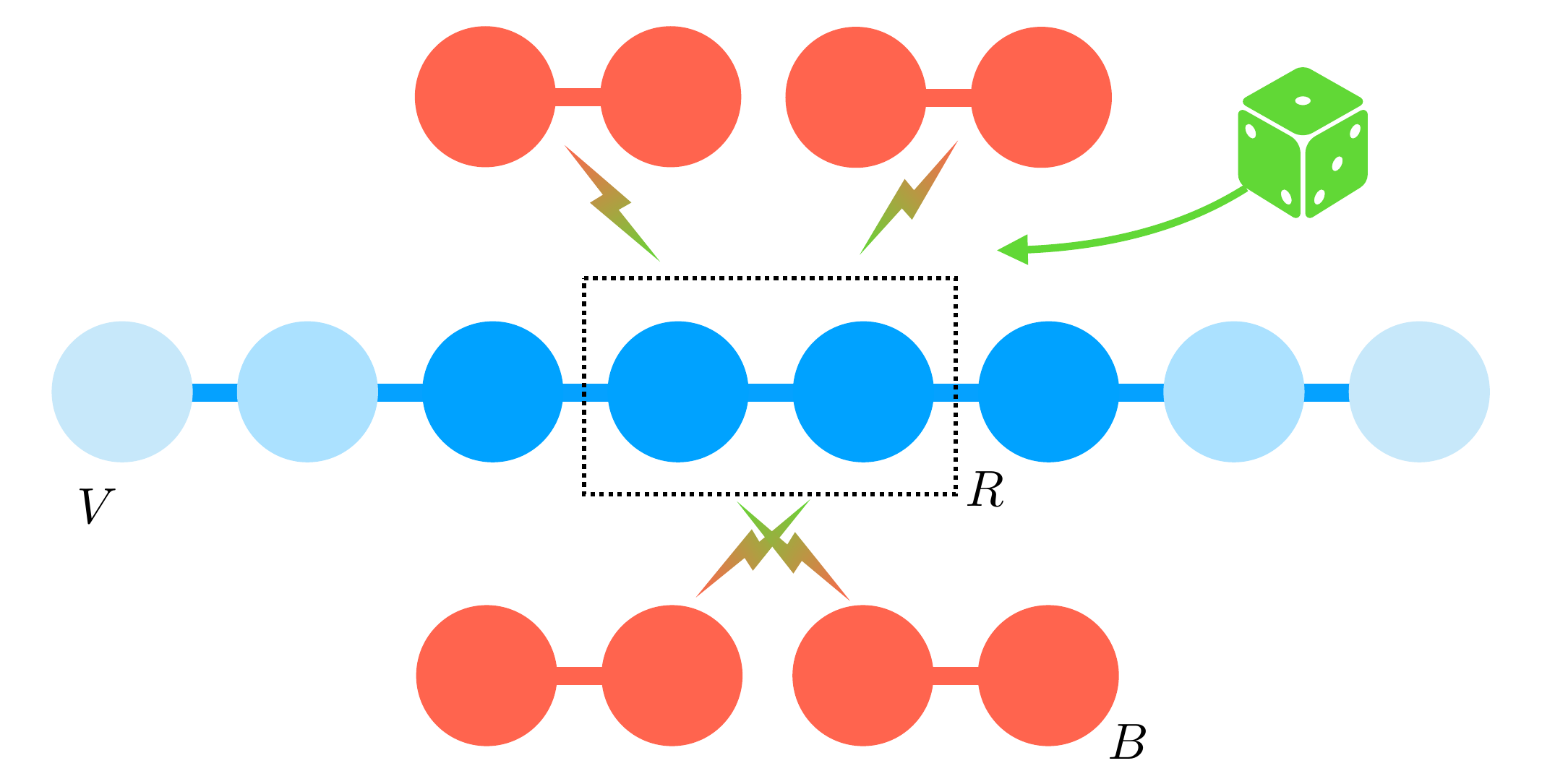}
\caption{{\bf Thermalization setting.} Thermalization of a region $R$ of an equilibrated lattice $V$ via stochastic collisions with an external bath $B$ (red). The collisions are modelled as randomly distributed, energy-preserving unitary interactions between $R$ and $B$, where interactions could be either between single subsystems of $B$ (top) or multiple ones (bottom). Our framework is inspired by the resource theoretic framework of thermal operations~\cite{janzing_thermodynamic_2000,brandao_resource_2013}, which have been used to investigate a wide variety of questions, such as the notion of work for microscopic systems~\cite{horodecki2013fundamental,brandao2015second}, quantum fluctuation theorems~\cite{alhambra2016fluctuating}, the third law~\cite{masanes_general_2017,wilming2017third}, and several other topics~\cite{woods2019maximum,halpern2019quantum,alhambra_revivals_2019}.}
	\label{fig:coupling}
\end{figure}
\par
We now turn to our model of the system-bath interactions, described via the following master equation,
\begin{equation}
\label{eq:master_eq}
\frac{\partial \, \rho_{RB}(t)}{\partial \, t}
=
\sum_{k} \frac{1}{r_k}
\left[
U_{RB}^{(k)} \, \rho_{RB}(t) \, U_{RB}^{(k) \dagger}
-
\rho_{RB}(t)
\right],
\end{equation}
where $\rho_{RB}$ is the global state on $R$ and $B$. This equation models a series of collisions, each described by a unitary operator $U_{RB}^{(k)} \in \BH{\mcH_R \otimes \mcH_B}$ acting non-trivially on the region $R$ and a subset of bath components, occurring at a given rate $r_k^{-1} > 0$ in time according to a Poissonian distribution (see Supplementary Note 1 for details). We consider elastic collisions, that is, we require that each unitary operator $U_{RB}^{(k)}$ conserves the global energy
\begin{equation}\label{eq:UHcommute}
[ U_{RB}^{(k)} , H_R + H_B ] = 0.
\end{equation}
It is worth noting that the above condition does not imply that the total energy of the lattice is conserved. In fact this is in general not the case, due to those operators $H_x$ in Eq.~\eqref{eq:lattice_ham}, with support on both $R$ and $R^c$. However, since by assumption these operator have support on at most $k$ adjacent lattice sites, if the region $R$ is sufficiently large one expects this boundary terms to contribute less and less to the total energy of the region. For high temperatures, such notions that boundary terms are negligible have rigorously been established \cite{kliesch_locality_2014}.
Note that Eq.~\eqref{eq:master_eq} is already in standard Lindblad form, and therefore describes a Markovian dynamical semi-group on $RB$. It is also important to note, however, that the process happening on the region $R$ is in principle non-Markovian, since $ B $ is modelled here explicitly, and can be a very small heat bath, which retains memory of the system's initial state. See Fig.~\ref{fig:coupling} for a graphical illustration of the setup.

In summary, the interaction model is general, in the sense that the unitary operators can act on a single subsystem (thus generating local Hamiltonian dynamics, for example), on two subsystems (these are the standard two-body interactions) or many more subsystems. Long-range interaction terms are also allowed, with the only restriction of Eq.~\eqref{eq:UHcommute} -- in the hope that one may work towards a relaxation in the future. We are nonetheless neglecting the interaction terms between the region $ R $ and $ R^c $. While this assumption is critical for allowing us to apply our framework to this problem, we note that there are situations where it is physically relevant -- for example, if the interaction between $ R $ and $ B $ occurs on a much shorter timescale than between $ R $ and $ R^c $, which is allowed by our class of interactions, as the rates $ \lbrace r_k \rbrace_k $ can be chosen to be high enough. Moreover, our results on bath size are independent on the system-bath coupling strength.
\par
We are finally in a position to define our central measure of robustness to thermalization. Let $\mathfrak{E}_{n}$ denote the set of quantum channels (i.e., completely-positive trace preserving maps acting over the quantum states of a system~\cite{nielsen_quantum_2010}) on $RB$ that can be generated via the above collision process for a bath of size $n-1$. Each channel in this set is realized through a different choice of unitary operations $\{ U_{RB}^{(k)} \}$, collisions rates $\{ r_k^{-1} \}$, and final time $t$. Given the region initial
state $\omega_R $, its Hamiltonian $ H_R $, and the bath Hamiltonian $ H_B $ defined as in Eq.~\eqref{eq:H_B_copiesofregion},
we define $n_\epsilon $ as the minimum integer such that there exists an element in $\mathfrak{E}_{n_\epsilon+1}$ that
$\epsilon$-thermalizes $R$,
\begin{equation}
n_\epsilon := \min \left\{ n \in \N \, \Bigg| \, \exists \, \chn \in \mathfrak{E}_{n+1}
\, : \,
\omega_R \xrightarrow[\epsilon, H_B]{\chn} \tau_{\beta}(H_R)\right\}  -1.
\end{equation}
The integer $n_\epsilon$ then quantifies the smallest size of a thermal environment required to thermalize region $R$ under
stochastic collisions, and hence provides a natural measure to quantify robustness of an MBL system against thermalization.
\subsection{Upper and lower bounds for thermalization}
Our main results are upper and lower bounds on $n_\epsilon$ that are essentially tight for a wide range of
Hamiltonians $H_R$. These bounds are stated using the (smooth) max-relative entropy, an entropic quantity
that has received considerable attention in recent years in quantum information and communication
theoretical research~\cite{konig2009operational,horodecki2013fundamental,bu2017maximum}. Once again, it is worth noting that our results are applicable in general for finite-dimensional quantum systems where equilibration occurs, of which MBL is a particularly interesting example.
The max-relative entropy~\cite{datta_min-_2009} between two quantum states $\rho, \sigma \in \SH{\mcH}$
such that $\supp(\rho)  \subseteq  \supp(\sigma)$ is defined as
\begin{equation}
\label{eq:max_rel_ent}
\Dmax{}\left( \rho \| \sigma \right) = \inf \left\{ \lambda \in \R \ : \ \rho \leq 2^{\lambda} \sigma \right\},
\end{equation}
while the smooth max-relative entropy between the same two quantum states, for some $\epsilon > 0$, is defined as 
 \begin{equation}
\label{eq:smooth_max_rel_ent}
\Dmax{\epsilon}(\rho \| \sigma) = \inf_{\tilde{\rho} \in B_{\epsilon}(\rho)} \Dmax{}\left( \tilde{\rho} \| \sigma \right),
\end{equation}
where $B_{\epsilon}(\rho)$ is the ball of radius $\epsilon$ around the state $\rho$ with respect to the distance induced
by the trace norm.
\subsubsection{Upper bound}
We first present and discuss the upper bound, which can be easily stated in terms of the quantities just introduced.
\begin{theorem}[Upper bound on the size of the bath] \label{thm:upper}
For a given Hamiltonian $H_R$, inverse temperature $\beta$, and a constant $\epsilon > 0$, we have that
\begin{align}
n_\epsilon \leq \frac{1}{\epsilon^2} \, 2^{\Dmax{}\left( \omega_R \| \tau_{\beta}(H_R) \right)}.
\end{align}
\end{theorem}
The above theorem provides a quantitative bound on the size of the thermal bath needed
to $\epsilon$-thermalize a lattice region $R$, when the coupling is mediated by stochastic collisions.
For this specific dynamics, the region can be $\epsilon$-thermalized if the size of the bath (the number of components)
is proportional to the exponential of the max-relative entropy between the state of the region $\omega_R$ and its
thermal state $\tau_{\beta}(H_R)$.
\par
Theorem~\ref{thm:upper} is proven in Supplementary Note 2. Here, we present a sketch of the proof
in two steps. In the first step, we show that $\mathfrak{E}_n$ can be connected to so-called random
unitary channels~\cite{audenaert_random_2008}. In the second step, we use this connection to find a
particular channel in $\mathfrak{E}_n$ that achieves the upper bound of the above theorem. A central
ingredient to the second step is a result known as convex split lemma~\cite{anshu_quantum_2017,
anshu_quantifying_2018}.
\par
Turning to the first step, recall that a random unitary channel is a map of the form
\begin{equation}
\label{eq:ran_unit_chn}
\chn(\cdot) = \sum_k \, p_k \, U_k \, \cdot \, U_k^{\dagger},
\end{equation}
where $\left\{ p_k \right\}_k$ is a probability distribution, and $\{U_k\}_k$ is a set of unitary operators. For a given number
$n-1$ of bath subsystems, we define the class of energy-preserving random unitary channels $\mathfrak{R}_n$
as those random unitary channels on $RB$ for which each unitary operator $U_k \in \BH{\mcH_R \otimes \mcH_B}$
commutes with the Hamiltonian of the global system, i.~e., $\left[ U_k , H_R + H_B \right] = 0$. In Supplementary Note 1,
we show that for any $n\geq1$, $\mathfrak{E}_n \subseteq \mathfrak{R}_n$, therefore allowing us to analyse any
element of $ \mathfrak{E}_n $ as a random unitary channel.
\par
Turning to the second step, we use a stochastic collision model with a simple representation in terms of random unitary
channels. Let us first recall that the thermal bath $B$ is described by $n-1$ copies of $\tau_{\beta}(H_R)$,
the Gibbs state of the Hamiltonian $H_R$ at inverse temperature $\beta$. The collisions occur either between the region $R$
and one subsystem of the bath, or between two bath subsystems. The rate of collisions is uniform, and given by $r^{-1} >0$.
During a collision involving the $i$-th and $j$-th subsystems of $RB$, the states of the two colliding components are
swapped, so that the interaction is described by the unitary operator $U^{(i,j)}_{\text{swap}}$. The action of this
operator over two quantum systems, described by the state vectors $\ket{\psi}_1$ and $\ket{\phi}_2$ respectively, is given by
$U^{(1,2)}_{\mathrm{swap}} \ket{\psi}_1 \otimes \ket{\phi}_2 =\ket{\phi}_1 \otimes \ket{\psi}_2$. For an initial global
state $\rho_{RB}^{(\text{in})} = \omega_R \otimes \tau_{\beta}(H_R)^{\otimes n-1}$, the steady state obtained through this
process is
\begin{equation}
\label{eq:equilibrium_state}
\rho_{RB}^{(\text{ss})}
=
\sum_{m=1}^n \frac{1}{n} \, \tau_{\beta}(H_R)^{\otimes m-1} \otimes \omega_R \otimes \tau_{\beta}(H_R)^{\otimes n-m},
\end{equation}
where the a-thermality of the region has been uniformly hidden into the different components of the bath. It is worth noting that, under the stochastic collision model described above, the global system reaches its steady state exponentially quickly in the collision rate $r^{-1}$, see Supplementary Note 2 for more details.
\par
The mapping from the initial state of region and bath to the steady state
is achieved by the following random unitary channel
\begin{equation}
\label{eq:equilibration_chn}
\bar{\chn}_n(\cdot) = \sum_{i=1}^n \frac{1}{n} \, U^{(1,i)}_{\text{swap}} \, \cdot \, U^{(1,i) \, \dagger}_{\text{swap}},
\end{equation}
which uniformly swaps each of the bath subsystems with $R$. Such channels have 
been studied before in the context of entropy production~\cite{diosi2006exact,csiszar2007limit}, see also
Ref.~\cite{scarani2002thermalizing} for a similar example. Since all subsystems share the same Hamiltonian $H_R$,
it is easy to see that each one of the $U^{(1,i)}_{\text{swap}}$ commutes with the joint Hamiltonian, so that $\bar{\chn}_n
\in \mathfrak{R}_n$. Finally, we can invoke the convex split lemma, see Supplementary Note 2, which allows us
to show that, for any $ \epsilon >0 $, the channel $\bar{\chn}_n$ can $\epsilon$-thermalize the
region $R$ when the number of subsystems is $n = \epsilon^{-2} 2^{\Dmax{}\left( \omega_R \| \tau_{\beta}(H_R) \right)}$.

The collision model presented here encompasses a wide range of physically realistic thermalization processes. However, before turning to a lower bound, we note that, due to the particularly simple nature of \eqref{eq:equilibration_chn}, one can use the above construction to upper bound the required size of a bath  for other thermalization models as well. For example, by noting that the channel~\eqref{eq:equilibration_chn} is permutation-symmetric (in the sense that it has permutation-invariant states as its fixed points), it follows that thermalization models with permutation-symmetric dynamics (i.e.~those allowing for any permutation symmetric channel) are also subject to the above upper bound.
\subsubsection{Lower bound and optimality results}
We now turn to deriving a lower bound on $n_\epsilon$. This bound is obtained through a further assumption on the
Hamiltonian $ H_R $, which we call the energy subspace condition (ESC).
\begin{definition}[Energy subspace condition]
Given a Hamiltonian $ H_R $, we say that it fulfills the ESC iff for any $n \in \N$,
given the set of energy levels $\left\{E_k \right\}_{k=1}^d$ of the Hamiltonian
$H_R$, we have that for any vectors $m, m' \in \N^d $ with the same normalization factor, namely
$\sum_k m_k = \sum_k m'_k = n $,
\begin{equation}
\label{item:opt_cond}
\sum_k m_k E_k \neq \sum_k m'_k E_k .
\end{equation}
\end{definition}

Let us here briefly discuss the physical significance of the ESC condition. The ESC entails (but is not equivalent to) that energy levels cannot be exact integer multiples of one another, which also implies full non-degeneracy. Furthermore, note that the ESC is not approximate, in other words, it still holds even if energy levels are very close to each other. 
	Having exact integer multiple energy levels is a very fine-tuned situation that breaks as soon as randomness is introduced in the Hamiltonian \cite{Linden2009}. Let us for example consider how likely it is for MBL systems to have degenerate energy levels. In the ergodic phase, the level statistics are Wigner-Dyson, therefore non-degeneracy is enforced by level repulsion. On the other hand, in the strong MBL phase, level statistics are Poissonian, meaning that the probability density function is maximum for zero-energy gaps. Despite so, the probability of exact degeneracy would correspond to a zero-volume integral of the Poissonian distribution (which is bounded from above), and therefore still amounts to zero probability of having degenerate gaps.

It is clear, however, that the ESC is much more stringent than requiring non-degeneracy; If we require Eq.~\eqref{item:opt_cond} to be satisfied for all $n \in \N$, this implies that the energy levels of the Hamiltonian	$H_R$ need to be irrational. Nevertheless, one can relax this condition by asking it to hold for all $n \leq N$, for some sufficiently large $N \in \N$, for example with the upper bound on $ n_\epsilon $ in Theorem \ref{thm:upper}. We refer to this as the ESC being satisfied up to $ N $. In Section \ref{sec:disorderedHeisenberg}, we discuss how a paradigmatic MBL system relates to this condition.

We can now state the following theorem, proved in Supplementary Note 4, on the optimality
of the channel associated with the convex split lemma.

\begin{theorem}[Optimal thermalization processes]
\label{thm:csl_optimality}
If $H_R$ satisfies the ESC and the state $\omega_R$ is diagonal in the energy eigenbasis, then the channel $\bar{\chn}_n$
in Eq.~\eqref{eq:equilibration_chn} provides the optimal thermalization process, that is, for any $n \in \N$,
\begin{equation}
\bar{\chn}_n \in \underset{\chn \in \mathfrak{E}_n}{\argmin}
\norm{ \chn \left( \omega_R \otimes \tau_{\beta}(H_R)^{\otimes n-1} \right) - \tau_{\beta}(H_R)^{\otimes n} }_1.
\end{equation}
\end{theorem}
Theorem~\ref{thm:csl_optimality} shows that, for Hamiltonians satisfying the ESC, the channel $\bar{\chn}_n$
provides the optimal thermalization of $R$, that is, no other random energy-preserving channel acting
on the same global system can achieve a smaller value of $\epsilon$ in Eq.~\eqref{eq:def_thermalizing}.
The above result applies to initial states that are diagonal in the energy eigenbasis; this is in general not the case for the reduced state $\omega_R$ of the infinite-time average of Eq.~\eqref{eq:inf_time_average}, since it might have coherence in the eigenbasis of the reduced Hamiltonian $H_R$. For states with coherence, the channel of Eq.~\eqref{eq:equilibration_chn} is not necessarily optimal anymore, but we can still bound the difference in thermalization achieved by this channel and an optimal one, see Supplementary Note 4 for the proof.
\begin{theorem}[Thermalization bound for coherent states]
\label{thm:csl_semi-optimality}
Fix $n \in \N$, and assume that $H_R$ satisfies the ESC. Consider the channel $\chn_{\mathrm{opt}} \in \mathfrak{E}_n$ achieving optimal thermalization $\epsilon_{\mathrm{opt}} = \norm{ \chn_{\mathrm{opt}} \left( \omega_R \otimes \tau_{\beta}(H_R)^{\otimes n-1} \right) - \tau_{\beta}(H_R)^{\otimes n} }_1$, and the decohering channel $\Delta( \cdot ) = \sum_E \Pi_E \cdot \Pi_E$, where $\Pi_E$ is the eigenprojector onto the energy subspace associated with $E$. We define the parameter $\delta = \norm{ \omega_R - \Delta(\omega_R) }_1$, quantifying the amount of coherence contained in the state of the region. Then, the thermalization achieved by the channel $\bar{\chn}_n$ is bounded as
\begin{equation}
\norm{ \bar{\chn}_n \left( \omega_R \otimes \tau_{\beta}(H_R)^{\otimes n-1} \right) - \tau_{\beta}(H_R)^{\otimes n} }_1
\leq
\epsilon_{\mathrm{opt}}
+
\delta.
\end{equation}
\end{theorem}
The above theorem provides a quantitative bound on the thermalization achieved by the channel $\bar{\chn}_n$ when the input system has coherence in the energy eigenbasis. In the case of MBL systems, the eigenstates of the Hamiltonian are close to product states, see for instance Ref~\cite{friesdorf_many-body_2015}, and therefore the reduced state of the infinite-time average $\omega_R$ is expected to have small and strongly-decaying coherence. Thus, Thm.~\ref{thm:csl_semi-optimality} shows that the stochastic collision model introduced in Section \ref{sec:setting} is able to effectively thermalize the region of an MBL systems.
From the two theorems stated above we derive the following corollary, providing a lower bound on the size of the thermal bath needed to thermalize a given quantum system.
\begin{corollary}[Lower bound to size of the bath]
For a given $\beta$ and $\epsilon > 0$, and some Hamiltonian $H_R$ satisfying the ESC, we have
\begin{equation}
\label{eq:lower_bound}
n_\epsilon \geq 2^{\Dmax{2 \sqrt{\epsilon} +\delta}\left( \omega_R \| \tau_{\beta}(H_R) \right)},
\end{equation}
where $\delta = \norm{\omega_R - \Delta(\omega_R)}_1$ and $\Delta(\omega_R)$ is the decohered version
of the state $\omega_R$.
\end{corollary}
Note that this lower bound is arising from the stringent 
model of thermalization of Eq.~(\ref{eq:def_thermalizing}), and that less stringent models will potentially lead
to smaller lower bounds. 
When $H_R$ does not satisfy the ESC, it is easy to find counter-examples to the optimality of the channel
$\bar{\chn}$, as we show in Supplementary Note 4. The idea is that this channel is optimal
only when it is able to produce a uniform distribution within each energy subspace of the global system,
and this is possible if each subspace is fully characterized by a different frequency of single-system eigenvectors,
which is exactly given by the ESC. Indeed, when the ESC is maximally violated,~i.~e., when the system Hamiltonian is completely degenerate, then no bath is required at all. See Sec.~III for a discussion of this and its relation to known bounds from randomness extraction.

\subsection{The disordered Heisenberg chain}
\label{sec:disorderedHeisenberg}
Our results from the previous section show that, for systems that satisfy the ESC, the max-relative entropy between the
local state of a lattice region and its thermal state provides a natural measure for the robustness of that region to thermalization
for a broad family of interactions. This includes many-body systems close to the transition between the ergodic and MBL
phase, where both level repulsion and randomness effects favour a lack of exact degeneracies, so that it seems reasonable
to expect for such systems to satisfy the ESC to sufficient order $n$. In this section, we use these results to study the robustness
of the MBL phase to the thermal noise for a concrete system. Specifically, we consider the disordered Heisenberg chain, a one-dimensional
spin-$\frac{1}{2}$ lattice system composed of $L$ sites, governed by the Hamiltonian 
\begin{equation}
\label{eq:XXZChainH}
H_V =
\sum_i^L \left( \sigma_i^x \sigma_{i+1}^x + \sigma_i^y \sigma_{i+1}^y+ \sigma_i^z \sigma_{i+1}^z \right)
+
\Delta \sum_i^L  h_i \sigma_i^z,
\end{equation}
where $\sigma^x, \sigma^y, \sigma^z \in \BH{\CC^2}$ are the Pauli operators, $\Delta$ is the (dimensionless) disorder
strength, and each parameter $h_i \in [-1,1]$ is drawn uniformly at random. We employ periodic boundary conditions.
\par
It has been demonstrated both
theoretically~\cite{BAA} and experimentally~\cite{Schreiber} that this system undergoes a localization transition
above the critical disorder strength $\Delta_c \approx 7$. The transition manifests itself in a breakdown of
conductance~\cite{BAA,Schreiber}, and a slowdown of entanglement growth after a quench~\cite{Prosen,Pollmann}.
Moreover, a phenomenological model in terms of quasi-local constants of motions exists which provides an explanation
for the non-thermal behavior of the system~\cite{Serbyn,Huse}.

To relate the model to our theoretical results, we note that when the region $ R $ is a single qubit ($|R|=1$) with a non-zero energy gap, then the ESC is always satisfied for all $ n\in \N $. For  $ |R|=2 $, we have verified the ESC condition across a large range of disorder strengths $ \Delta \in [0.01,20] $, up to $N=25$ (out of $\sim 2000$ realizations, all of them satisfy the ESC); we have additionally considered a small number of realizations for $N=35$, for all of which the ESC holds. As $ |R| $ increases, higher values of $ N $ become significantly harder to check numerically. However, we have verified that for example $|R|=4$ always satisfies the ESC up to $N=5$ (for 2000 generated realizations), while the condition is satisfied with high probability when $N=6$ ($90\%$ of the realizations).
\begin{figure}[t]
\center
\includegraphics[width=0.9\columnwidth]{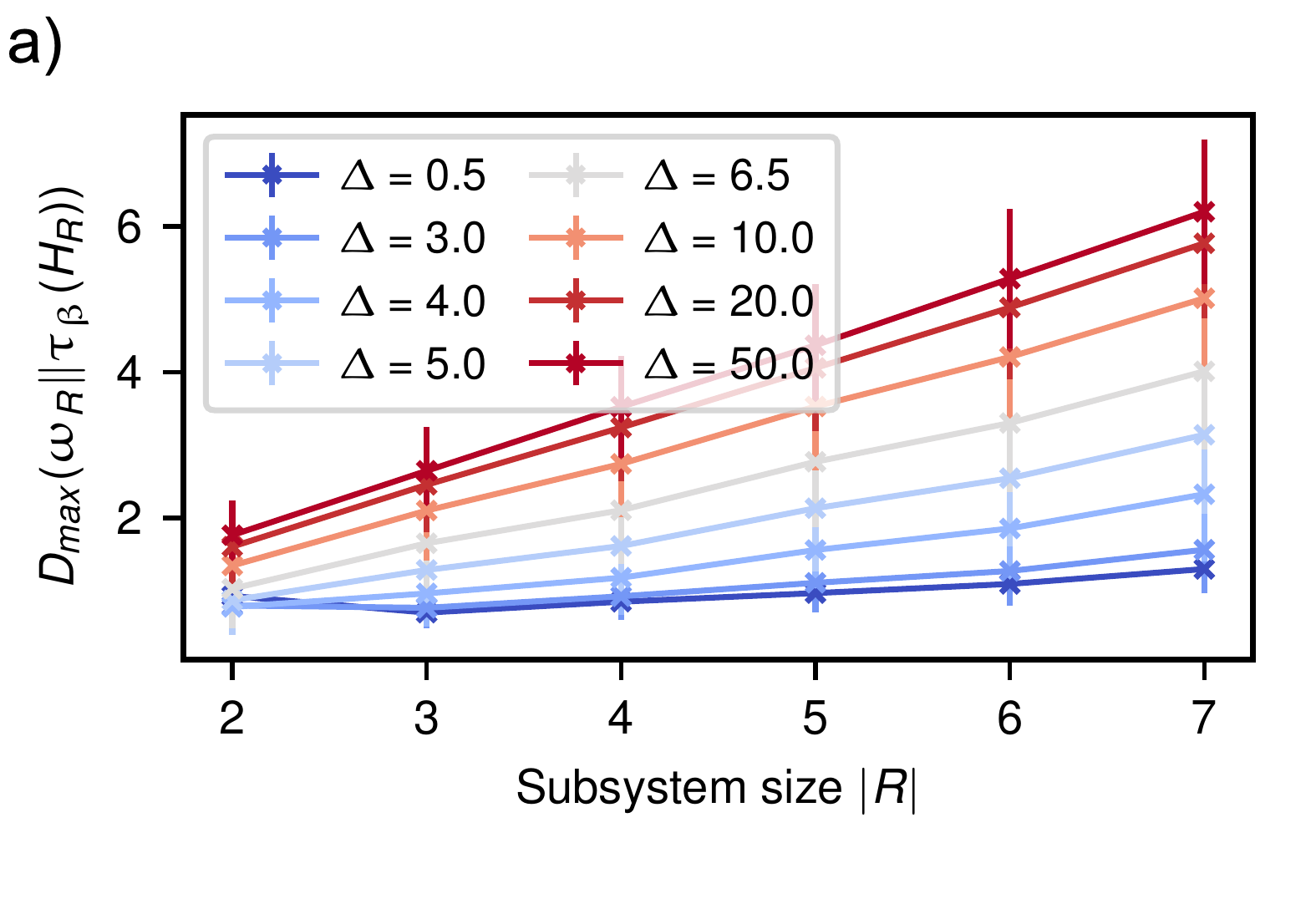}
\includegraphics[width=0.9\columnwidth]{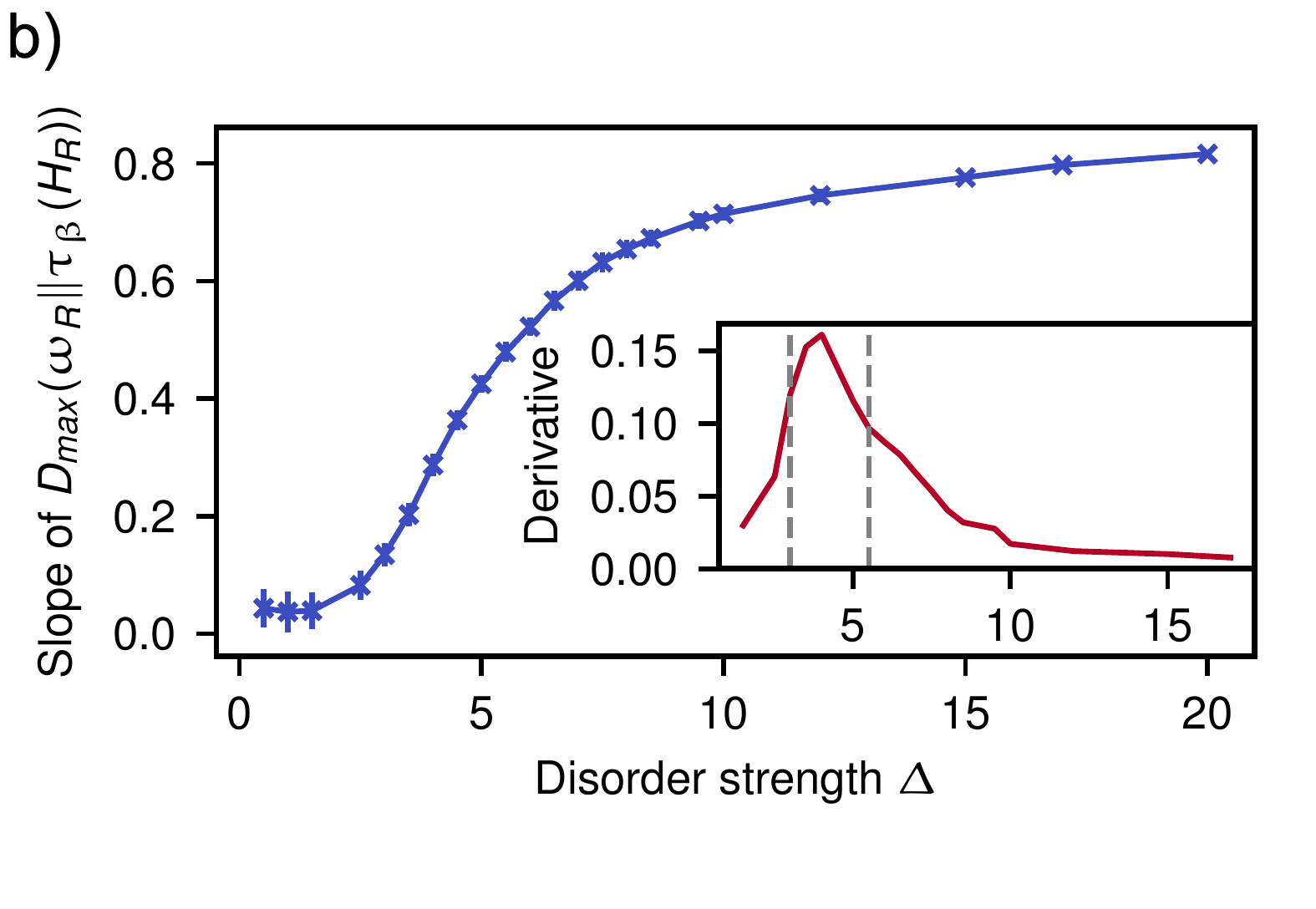}
\caption{{\bf Max-relative entropy for the disordered Heisenberg chain.} {\bf a)} Max-relative entropy $\Dmax{} \left( \omega_R \| \tau_{\beta}(H_R) \right)$ as a function
of subsystem size $|R|$ for a lattice of $L=15$ sites. The plots show an average over $100$ disorder
realizations. The states were calculated employing exact diagonalization. For low values of disorder
$\Delta$ the max-relative entropy is almost constant as $|R|$ increases, while for higher values of
$\Delta$ it scales linearly in $|R|$, hinting toward a robustness of the MBL phase with respect to the
class of interaction models we are considering. {\bf b)} The slope of the max-relative entropy as
a function of the disorder $\Delta$ provides information on the phase transition. Indeed, we can see that this
quantity abruptly increases in proximity of the expected phase transition from the ergodic to the MBL
phase. The slope is obtained by a linear fit with error bars indicating least squares errors. The inset
shows the derivative of the slope, with the grey lines indicate a possible transition region.
}
\label{fig:Dmax}
\end{figure}
\par
In our simulation, we choose as initial state vector $\ket{\Psi(0)}$ a variation of the N{\'e}el state
with support on the total-magnetization sectors $M=\pm 1, 0$. Our choice is motivated by the fact that
this state, due to its increased overlap over different symmetric subspaces of the Hamiltonian, thermalizes
more easily during the ergodic phase. For each random realization, we numerically compute the infinite
time average of $\ket{\Psi}$ as defined in Eq.~\eqref{eq:inf_time_average}, using exact diagonalization.
We then trace out part of the lattice so as to obtain the state $\omega_R$, describing 
the infinite time averaged state reduced to the region $R$.
Notice that in the ergodic phase, when the disorder strength 
$\Delta < \Delta_c$, this state is expected to be close to thermal, with a temperature
depending on the energy of the initial state of the lattice. 
However, when the disorder strength $\Delta$ passes its
critical value, the state $\omega_R$ is not thermal any more \cite{Schreiber}.
\par
To numerically compute the max-relative entropy for this system, we use the Gibbs state
of the reduced Hamiltonian $H_R$, obtained from the Hamiltonian in Eq.~\eqref{eq:XXZChainH}
by only considering terms with full support on the region $R$. The inverse temperature $\beta$ is obtained
by constructing the global Gibbs state of the lattice and requiring its energy to equal to the one
of the initial state vector $\ket{\Psi(0)}$. We compute $\Dmax{}(\omega_R\| \tau_{\beta}(H_R))$ for
different disorder strengths $\Delta$, and different sizes of the region $R$, see Fig~\ref{fig:Dmax}.
\par
We find that in the ergodic phase the state is approximately thermal, and the max-relative
entropy remains almost constant as $|R|$ increases. For big enough sizes of the region, the max-relative
entropy starts increasing even in the ergodic case. However, this effect is due to the finite size of the lattice in our simulation,
and it can be mitigated by increasing the number of lattice sites (at the expenses of a higher computational cost).
As $\Delta$ approaches the critical value, we find that the max-relative entropy scales linearly in the region
size $|R|$, with a linear coefficient which increases with the disorder strength, see Fig.~\ref{fig:Dmax}.a.
As a result, the size of the external thermal bath $n_{\epsilon}$ scales exponentially in the region size, due to the
bounds we have obtained in the previous section. This exponential scaling in the size of the bath
suggests a robustness of the MBL phase with respect to the dynamics given by Eq.~\eqref{eq:master_eq}, since the
relative size of the bath ${n_{\epsilon}}/{|R|}$ needs to diverge as $|R|$ tends to infinity. In other words, for the
MBL phase to be destroyed one needs, under the interaction models we consider, an exponentially vast amount of
thermal noise. It is worth noting that our characterization of robustness of the MBL phase to thermal noise is
distinctly different from others found in the literature~\cite{fischer_dynamics_2016,nandkishore_spectral_2014,
johri_many-body_2015,levi_robustness_2016}. Indeed, we couple the system with a finite-sized thermal bath,
and quantify the robustness in terms of its size. Furthermore, our notion of thermalization accounts for the
evolution of both system and bath, rather than focusing on the system only. Other works instead
consider infinite thermal reservoirs and quantify the robustness as a function, for instance, of the coupling
between system and environment. A promising realization are recent optical lattice experiments
\cite{bordia2016coupling,bordia2017probing,abadal2018bath}. However, to connect to our findings one would need
full state tomography on both system and bath which so far is out of reach for these platforms.
\par
We additionally study the first derivative of the max-relative entropy with respect to the region size $|R|$, as
a function of disorder strength, shown in Fig.~\ref{fig:Dmax}.b. 
We find that, during the ergodic phase, the derivative remains constant
and small. As $\Delta$ approaches the critical value, the derivative increases, and for $\Delta \gg \Delta_c$
the derivative becomes constant again. Thus, we find that
the derivative of the max-relative entropy with respect to the region size is an order parameter for the MBL
phase transition. We then use this order parameter to estimate the critical value $\Delta^{(L)}_c$ for the
finite-length spin chain we are considering, obtaining a value of approximatively $4.5$ for $L = 15$ sites.
While the critical value for infinite-length spin chains is considered to be $\Delta_c \approx 7$, we find that
our value, which we stress is obtained for a finite number of sites, seems to be in good accord with known
results found in the literature using other measures~\cite{luitz_many-body_2015,gray_many-body_2018}.
\section{Discussion}
We show that mathematical results originally developed to study quantum information processing may find
their applications in many body physics, in particular for the study of MBL in this paper. We demonstrate
this by applying the recently developed convex split lemma technique, to derive upper and lower bounds for
the size of the external thermal bath required to thermalize an MBL system. The class of interaction models
between lattice and thermal bath is described by the master equation~\eqref{eq:master_eq}, and 
consists in stochastic energy-preserving collisions between the system and bath components. The
bounds we obtain depend on the max-relative entropy between the state we aim to thermalize and its thermal
state.
\par
We make use of these analytic results to study a specific and at the same time much ubiquitous 
system exhibiting MBL features, known 
as the disordered Heisenberg chain. We show
that the MBL phase in this system is in fact robust with respect to the thermalization processes considered 
here, and that
the derivative of the max-relative entropy with respect to region size serves as an approximate order parameter of the
ergodic to MBL transition. We emphasize
that this is not in contradiction with previous results, where a breakdown of localization was 
reported~\cite{fischer_dynamics_2016,
nandkishore_spectral_2014,levi_robustness_2016}, as the size of the baths considered in these works was
unbounded. Resource-theoretic frameworks offer another potentially useful approach for studying thermalization
with infinite-dimensional baths; the framework of elementary thermal operations~\cite{lostaglio2018elementary}
which involves a single bosonic bath that is coupled only with two levels of the system of interest. One may
then study the resources (the number of bosonic baths with different frequencies) required to achieve
thermalization. Also, and more technically, it would be interesting to study the extent to which both the ESC condition and the requirement of exact commutation in our framework can be relaxed to only approximately hold true and how this in turn affects the lower bound of Corollary~\ref{eq:lower_bound}. These questions we leave to be studied in future work.
\par
The success of our application implies that, potentially, other information theoretic tools could be
employed to study the thermalization of MBL systems -- and non-equilibrium dynamics of many-body systems
in more generality, for that matter. For instance, results in randomness
extraction~\cite{trevisan_extracting_2000} might be useful to provide new bounds. In randomness extraction,
a weakly random source is converted into an approximately uniform distribution, with the use of a seed (a small,
uniformly distributed auxiliary system). In analogy, thermalization requires a non-thermal state to be mapped into an almost
thermal state, with the help of an external bath (the seed). Thus, it seems possible that results from randomness
extraction might be modified to study this setting, and to obtain bounds on the thermal seed.
\par
It has been shown that excited states of one-dimensional MBL systems are well-approximated by
matrix product states (MPS) with a low bond dimension~\cite{friesdorf_many-body_2015,Bauer} 
if the system features an information mobility gap. These states have several
interesting properties, and in particular they feature an area law for the entanglement entropy which is logarithmic in the bond
dimension~\cite{eisert_colloquium:_2010}. Since our result is based on a particular entropic quantity, it might
be possible to use the properties of MPS to derive a fully analytical bound on the robustness of these systems
with respect to thermal noise. It is the hope that our work stimulates further cross-fertilization between the fields
of quantum thermodynamics and the study of quantum many-body systems out of equilibrium.
\section{Acknowledgments}

We would like to thank {\'A}lvaro Martin Alhambra, Johnnie Gray, Volkher Scholz, and Henrik Wilming for helpful discussions. C.~S.\ is funded by EPSRC. N.~N. is funded by the Alexander von Humboldt foundation. P.~B. acknowledges funding from the Templeton Foundation. J.~E. has been supported by the DFG (FOR 2724, CRC 183), the FQXi, and the Templeton Foundation. We acknowledge the Freie Universit{\"a}t Berlin for covering the costs of offsetting the emission generated by this research.

\vspace{0.05cm}

\begin{figure}[H]
	\begin{center}
		\begin{tabular}[b]{l c}
			\hline
			\textbf{Numerical simulations} & \\
			Total Kernel Hours [$\mathrm{h}$]& 120000\\
			Thermal Design Power per Kernel [$\mathrm{W}$]& 5.75\\
			Total Energy Consumption of Simulations [$\mathrm{kWh}$] & 1960\\
			Average Emission of CO$_2$ in Germany [$\mathrm{kg/kWh}$]& 0.56\\
			Total \co-Emission from Numerical Simulations [$\mathrm{kg}$] & 1098\\
			Were the Emissions Offset? & \textbf{Yes}\\
			\textbf{Transportation} & \\
			Total \co-Emission from Transportation [$\mathrm{kg}$] & 2780 \\
			Were the Emissions Offset? & \textbf{Yes}\\
			\hline
			Total \co-Emission [$\mathrm{kg}$] & 3878\\
			\hline
		\end{tabular}
		\caption{Estimated climate footprint of this paper. Prototyping is not included in these calculations. Estimations have been calculated using the examples of Scientific CO$_2$nduct \cite{scicon2019} and are correct to the best of our knowledge.}
	\end{center}
\end{figure}
%
%

\newpage
\onecolumngrid
\title{Supplementary Information: Bounding the resources for thermalizing many-body localized systems}
\maketitle
\section*{Supplementary Note 1: Collisional master equation and random unitary channels}
In this section, we consider a large class of dynamical processes involving $n$ subsystems (for example, particles or molecules),
where each of the subsystems has a given probability in time of interacting with another one (or with more than one at a time),
and the interaction is fully general as long as it conserves the total energy. For simplicity, below we only consider two-body
interactions, but extending the setting to $k$-body interactions is straightforward. Any such process can be described by
the following master equation,
\begin{equation}
\label{eq:stoch_process}
\frac{\partial \, \rho_n(t)}{\partial \, t} = \sum_{i, j} \lambda_{i,j} \left[ U_{i,j} \, \rho_n(t) \, U_{i,j}^{\dagger} - \rho_n(t) \right],
\end{equation}
where $\rho_n(t)$ is the state of the $n$ subsystems at time $t$, and $\lambda_{i,j} > 0$ is the rate at which the interaction $U_{i,j}$
between the $i$-th and $j$-th subsystem occurs. Each interaction is energy preserving, that is, $\left[ U_{i,j}, H^{(n)} \right] = 0$,
where $H^{(n)}$ is the Hamiltonian of the global system. In the case in which only two-body interactions are present in the
above equation, the total number of different unitaries is $N = \frac{1}{2} n (n-1)$. In the following, we re-label the two indices
$i,j$ with a single one $k = i,j$ taking values between $1$ and $N$.
At any given time $t > 0$, we show in the next section that the solution of the above process has the form
\begin{equation}
\rho_n(t) = \sum_{\mathbf{m}\in\mathbb{Z}^N} P_t(\mathbf{m}) \cdot \rho(\mathbf{m}) ,
\end{equation}
where the distribution
\begin{equation}
P_t(\mathbf{m}) = \prod_{k=1}^N p^{(\lambda_k)}_t(m_k),
\end{equation}
is given by the product of $N$ Poisson distributions, each with a different mean value $ \lambda_k t $.  For a given
$\mathbf{m} = \left( m_1 , m_2, \ldots, m_N \right)^T$ such that $m = \sum_{k=1}^N m_k$, the value of $P_t(\mathbf{m})$ is
the probability that $m$ interactions occurs during the time $t$, of which $m_k$ are described by the $k$-th unitary
operator $U_k$. On the other hand, the state $\rho(\mathbf{m})$ is a uniform mixture of different states obtained from
the initial state $\rho_n(t=0)$, by considering all different sequences of events giving rise to the same $ \mathbf{m} $,
\begin{equation}
\label{eq:m_state}
\rho(\mathbf{m}) = \frac{1}{\mathcal{N}_{\mathbf{m}}}
\sum_{\pi} U_{\pi}^{\mathbf{m}} \, \rho_n(t=0) \, (U_{\pi}^{\mathbf{m}})^{\dagger},
\end{equation}
where the unitary operator is
\begin{equation}
\label{eq:unitary_perm}
U_{\pi}^{\mathbf{m}} =
\pi(
\underbrace{U_1 \ldots U_1}_{m_1}
\underbrace{U_2 \ldots U_2}_{m_2}
\ldots
\underbrace{U_N \ldots U_N}_{m_N}
),
\end{equation}
and $\pi$ is an element of the symmetric group $S_m$ which produces a different permutation of the $m$ unitaries
describing the two-body interactions. Furthermore, 
\begin{equation}
\mathcal{N}_{\mathbf{m}} := \binom{m}{m_1,\cdots,m_N}
\end{equation} is the
multinomial coefficient, which is precisely the number of different permutations. Notice that the unitary operator
$U_{\pi}^{\mathbf{m}}$ commutes with the total Hamiltonian $H^{(n)}$, since it is obtained from two-body interactions
which commutes with $H^{(n)}$. In the special case in which the two-body operators $\left\{ U_{k} \right\}_k$ commute
with each other, the permutation $\pi$ can be dropped, and in Eq.~\eqref{eq:m_state} we can remove the weighted sum.
\par
As a result, the evolution of the system at any finite time $t\geq0$ can be described by a convex mixture of energy-preserving
unitary operations; specifically, for the dynamics given by the master equation~\ref{eq:stoch_process}, the channel which
maps the initial state into $\rho_n(t)$ is given by
\begin{equation}
\mathcal{E}(\cdot) = \sum_{\mathbf{m}\in\mathbb{Z}^N}
\frac{P_t(\mathbf{m})}{\mathcal{N}_{\mathbf{m}}}
\sum_{\pi}
(U_{\pi}^{\mathbf{m}}) \, \cdot \, (U_{\pi}^{\mathbf{m}})^{\dagger}.
\end{equation}
We can formally define a class of channels which are associated with this dynamics.
\begin{definition}[Convex mixtures of energy-preserving unitary operations]
	\label{def:ruc}
	For a given number of subsystems $n \in \N$ and a global Hamiltonian $H^{(n)}$, we define the class of maps $\mathfrak{E}_n$
	as composed by every channel generated by the master equation~\eqref{eq:stoch_process}, for any finite time $t \geq 0$,
	and for any choice of rates $\left\{ \lambda_{i,j} \right\}$ and of energy-preserving unitary operations $\left\{ U_{i,j} \right\}$.
\end{definition}
This very general set of maps is at the core of the thermalization results we derive in Supplementary Note 2.
\subsection{Solution of a non-commuting Poisson process}
\begin{proposition}[Solution to non-commuting Poisson process]
	\label{prop:solution_poisson_process}
	Consider a finite-dimensional Hilbert space $\mcH$ and the following master equation
	\begin{equation}
	\label{eq:poisson_process}
	\frac{\partial \, \rho(t)}{\partial \, t} = \sum_{k=1}^N \lambda_k \left[ U_k \rho(t) U_k^{\dagger} - \rho(t) \right],
	\end{equation}
	where $\rho(t) \in \SH{\mcH}$, and for all $k$ the coefficient $\lambda_k > 0$, and $U_k$ is a unitary operator
	acting on $\mcH$. The solution of this master equation is
	\begin{equation}
	\label{eq:solution_pp}
	\rho(t) = \sum_{m=0}^{\infty} p_t^{(\bar\lambda)}(m) \, \rho(m),
	\end{equation}
	where $p_t^{(\bar{\lambda})}(m)$ is a Poisson distribution whose mean value is $\bar{\lambda} \, t$, with $\bar{\lambda} = \sum_k \lambda_k$.
	The state $\rho(m)$ is obtained from the initial state $\rho(t=0)$ by applying $m$ times the channel $\Lambda$, i.e. $\rho(m) = \Lambda^{(m)} (\rho(t=0))$, where
	\begin{equation}
	\label{eq:lambda_map}
	\Lambda(\cdot) = \sum_{k=1}^N \frac{\lambda_k}{\bar{\lambda}} U_k \, \cdot \, U_k^{\dagger}.
	\end{equation}
\end{proposition}
\begin{proof}
	Let us first rearrange Eq.~\eqref{eq:poisson_process} in such a way that, on the right-hand side, only one operator is acting on
	the state $\rho(t)$, that is
	\begin{align*}
	\frac{\partial \, \rho(t)}{\partial \, t}
	&=
	\bar{\lambda} \left(\sum_{k=1}^N \frac{\lambda_k}{\bar{\lambda}} U_k \rho(t) U_k^{\dagger} - \rho(t) \right) \\
	&=
	\bar{\lambda} \big( \Lambda( \rho(t) ) - \rho(t) \big),
	\end{align*}
	where the map $\Lambda$ is defined in Eq.~\eqref{eq:lambda_map}. To show that Eq.~\eqref{eq:solution_pp} is the
	solution of the above differential equation, first notice that the Poisson distribution $p_t(m)$ is the sole time-dependant
	object, since $\rho(m) = \Lambda^{(m)}(\rho_0) = \Lambda \circ \ldots \circ \Lambda (\rho_0)$, and $\rho_0 = \rho(t=0)$.
	Thus, when taking the time derivative of $\rho(t)$, we can exploit the fact that
	\begin{equation*}
	\frac{\partial \, p_t(m)}{\partial \, t} = \bar{\lambda} \left[ p_t(m-1) - p_t(m) \right] .
	\end{equation*}
	Then, the time derivative of the solution in Eq.~\eqref{eq:solution_pp} is given by
	\begin{align*}
	\frac{\partial \, \rho(t)}{\partial \, t}
	&=
	\sum_{m=0}^{\infty} \bar{\lambda} \left[( p_t(m-1) - p_t(m) \right] \Lambda^{(m)}(\rho_0)
	=
	\bar{\lambda}
	\left[
	\sum_{m=0}^{\infty}  p_t(m) \Lambda^{(m+1)}(\rho_0)
	-
	\sum_{m=0}^{\infty} p_t(m) \Lambda^{(m)}(\rho_0)
	\right] \\
	&=  \bar{\lambda}
	\left[
	\Lambda \left( \sum_{m=0}^{\infty}  p_t(m) \rho(m) \right)
	-
	\sum_{m=0}^{\infty} p_t(m) \rho(m)
	\right]
	=
	\bar{\lambda} \big( \Lambda( \rho(t) ) - \rho(t) \big),
	\end{align*}
	where in the third line we use the fact that $\Lambda$ is linear and continuous.
\end{proof}
A straightforward corollary of the above proposition concerns the relation between the maps in $\mathfrak{E}_n$
and the class of (energy-preserving) \emph{random unitary channels}~\cite{audenaert_random_2008}, defined as
follows.
\begin{definition}[Energy preserving random unitary channels]
	For a given number of subsystems $n \in \N$ and a global Hamiltonian $H^{(n)}$, we define the class of energy-preserving
	random unitary channels $\mathfrak{R}_n$ as composed by every maps of the form
	\begin{equation}
	\chn( \cdot ) = \sum_k p_k \, U_k \cdot U_k^{\dagger},
	\end{equation}
	where $\left\{ p_k \right\}_k$ is a probability distribution, and each unitary operator $U_k$ preserves the energy,
	that is, $\left[ U_k , H^{(n)} \right] = 0$.
\end{definition}
The corollary of Prop.~\ref{prop:solution_poisson_process} is then the following.
\begin{corollary}[$\mathfrak{E}_n$ as subsets of energy preserving random unitary channels]
	\label{cor:E_sub_R}
	Given number of subsystems $n \in \N$ and a global Hamiltonian $H^{(n)}$, the set of maps $\mathfrak{E}_n$ is
	a subset of the class of energy-preserving random unitary channels $\mathfrak{R}_n$.
\end{corollary}
\subsection{Solution of several independent Poisson processes}
Notice that the solution of Eq.~\eqref{eq:poisson_process} can be modified so that the single Poisson distribution in
Eq.~\eqref{eq:solution_pp} is replaced by a product of $N$ Poisson distributions, each of them governing the
number of a specific two-body interaction applied to the initial state at time $t$. To show this, let us first introduce
some useful notation. For the $ i $-th two-body process, denote $ \lambda_i $ as the corresponding rate, and denote
the tuple $ \mathbf{\lambda} = (\lambda_1,\cdots,\lambda_N) $. Consider the state $\rho(m) = \Lambda^{(m)}(\rho_0)$,
which is the state where a total of $ m $ such two-body interactions have happened. Let $ \mathcal{M}_m = \lbrace
\mathbf{m}\in\mathbb{N}^N |\sum_i m_i = m \rbrace $ denote the set of $ N $-dimensional tuples consisting of non-negative
integers, such that the sum of all elements equals $ m $. We can explicitly re-write $\rho(m)$ as
\begin{align}\label{eq:rho_m}
\rho(m) = \Lambda^{(m)}(\rho_0) = \sum_{\mathbf{m}\in\mathcal{M}_m} \mathcal{N}_{\mathbf{m}}\cdot C_{\vec\lambda}\cdot \rho(\mathbf{m}),
\end{align}
where $\mathcal{N}_{\mathbf{m}} := \binom{m}{m_1, \ldots, m_N}$ is the multinomial coefficient, 
\begin{align}
C_{\vec\lambda} :=
\prod_{k=1}^N \left( \frac{\lambda_k}{\bar{\lambda}} \right)^{m_k}  
\end{align}
is the product of the corresponding weights,
and the state $\rho(\mathbf{m})$ is a mixture over all possible combination of $m$ unitaries, where each unitary
$U_k$ appears $m_k$ times,
\begin{equation}
\rho(\mathbf{m}) =
\frac{1}{\mathcal{N}_{\mathbf{m}}}
\sum_{\pi} U_{\pi}^{\mathbf{m}} \, \rho_0 \, (U_{\pi}^{\mathbf{m}})^{\dagger},
\end{equation}
where the unitary $U_{\pi}^{\mathbf{m}}$ has been defined in Eq.~\eqref{eq:unitary_perm}.
\par
Let us now consider the corresponding Poisson distribution $P_t^{(\bar\lambda)}(m)$, with a mean value $ \bar\lambda t$.
We want to see how this relates to the $ N $ independent Poisson processes with mean values $ \vec\lambda $. 
Noting that $ \sum_k m_k = m $ and $ \sum_k \lambda_k = \bar\lambda $, we can re-write this distribution as follows,
\begin{align}\label{eq:joint_prob_dist}
p_t^{(\bar\lambda)} (m) &= \left( \bar{\lambda} t \right)^m  \frac{e^{- \bar{\lambda} t}}{m!}
= \frac{1}{m!} \prod_{k=1}^N \left( \bar{\lambda} t \right)^{m_k} e^{- \lambda_k \, t}
= \frac{1}{\mathcal{N}_{\mathbf{m}}}\cdot \prod_{k=1}^N \frac{\left( \bar{\lambda} t \right)^{m_k} e^{- \lambda_k \, t}}{m_k !} \nonumber\
=\frac{1}{\mathcal{N}_{\mathbf{m}}} \cdot \prod_{k=1}^N \left[ \left( \frac{\bar{\lambda}}{\lambda_k} \right)^{m_k} p^{\lambda_k}_t(m_k)\right] \nonumber\\
&=\frac{1}{\mathcal{N}_{\mathbf{m}}}\cdot \frac{1}{C_{\vec\lambda}}\cdot \prod_{k=1}^N p^{\lambda_k}_t(m_k),
\end{align}
by noting that each $p^{(\lambda_k)}_t(m_k)$ is a Poisson distribution with mean value $\lambda_k \, t$. If we now
replace Eq.~\eqref{eq:rho_m} and Eq.~\eqref{eq:joint_prob_dist} into Eq.~\eqref{eq:solution_pp}, we see that the
coefficients $ \mathcal{N}_{\mathbf{m}} $ and $  C_{\vec\lambda}$ cancel out, and therefore, for a given time $t > 0$,
\begin{align}
\rho(t) &= \sum_{m=0}^{\infty} \sum_{\mathbf{m}\in\mathcal{M}_m}
\left[ \prod_k p^{(k)}_t(m_k) \right] \rho(\mathbf{m})
= \sum_{\mathbf{m}\in\mathbb{Z}^N}
\left[ \prod_k p^{(\lambda_k)}_t(m_k) \right] \rho(\mathbf{m}).
\end{align}
%

\section{Supplementary Note 2: Upper and lower bounds on the size of the thermal environment}
In this section we prove the main results presented in the main text, namely the upper and lower bounds
to the size of the external thermal bath required to thermalize a region of a many-body system. These
bounds depend on the entropic quantity known as \emph{max-relative entropy}~\cite{datta_min-_2009},
defined for two quantum states $\rho, \sigma \in \SH{\mcH}$ such that $\supp(\rho) \subseteq \supp(\sigma)$
as
\begin{equation}
\Dmax{}\left( \rho \| \sigma \right) = \inf \left\{ \lambda \in \R \ : \ \rho \leq 2^{\lambda} \sigma \right\}.
\end{equation}
The smooth max-relative entropy between the same two quantum states, for $\epsilon > 0$, is defined as 
\begin{equation}
\label{eq_supp:smooth_max_rel_ent}
\Dmax{\epsilon}(\rho \| \sigma) = \inf_{\tilde{\rho} \in B_{\varepsilon}(\rho)} \Dmax{}\left( \tilde{\rho} \| \sigma \right),
\end{equation}
where $B_{\epsilon}(\rho)$ is the ball of radius $\epsilon$ around the state $\rho$ with respect to the distance induced
by the trace norm.
The main technical tool we use to derive our bounds is a result from quantum information theory known as
the \emph{convex split lemma}, first introduced and proved in Ref.~\cite[Lemma~2.1 in Supp.~Mat.]{anshu_quantum_2017}.
\begin{lemma}[Convex split lemma]
	\label{lem:csl}
	Consider a finite-dimensional Hilbert space $\mcH$ and two states $\rho, \sigma \in \SH{\mcH}$ such that
	$\supp(\rho) \subseteq \supp(\sigma)$. Then the state $\rho^{(n)} \in \SH{\mcH^{\otimes n}}$ defined as
	\begin{equation}
	\label{eq:csl_state}
	\rho^{(n)}
	=
	\frac{1}{n} \sum_{m=1}^n
	\sigma^{\otimes m-1} \otimes \rho \otimes \sigma^{\otimes n-m},
	\end{equation}
	is such that its trace distance to the $n$-copy i.i.d state $\sigma^{\otimes n}$ is upper-bounded as
	\begin{equation}
	\norm{\rho^{(n)} - \sigma^{\otimes n}}_1^2 \leq {\frac{2^{\Dmax{}\left( \rho \| \sigma \right)}}{n}}.
	\end{equation}
\end{lemma}
In the following, we consider a specific channel which, when acting on the $n$-subsystem state $\rho \otimes
\sigma^{\otimes n-1}$, is able to produce the state $\rho^{(n)}$ given in Eq.~\eqref{eq:csl_state} of the above
lemma. This channel belongs to the set of random unitary channels acting on $n$ subsystems, and is defined as
\begin{equation}
\label{eq:channel_csl}
\bar{\chn}_n(\cdot) = \frac{1}{n} \sum_{i=1}^n U^{(1,i)}_{\text{swap}} \, \cdot \, U^{(1,i) \, \dagger}_{\text{swap}},
\end{equation}
where $U^{(i,j)}_{\text{swap}}$ denotes the unitary swap between the $i$-th and the $j$-th subsystems.
\par
We now specialise the setting to the one considered in the main text. We consider a region $R$ described
by the Hilbert space $\mcH_R$, with Hamiltonian $H_R$ and state $\omega_R$, and an external bath $B$
composed by $n-1$ subsystems at inverse temperature $\beta$. The Hilbert space of the bath is $\mcH_B
= \mcH_R^{\otimes n-1}$, with Hamiltonian $H_B = \sum_{i=1}^{n-1} H_R^{(i)}$, where the operator $H_R^{(i)}$
only acts non-trivially on the $i$-th subsystem of the bath. The state of the bath is thermal, thus defined by
$\tau_{\beta}(H_B) = \tau_{\beta}(H_R)^{\otimes n-1}$. We are interested in the process of thermalization of
the region $R$ by means of the collisional models described by the master equation of the form given in
Eq.~\eqref{eq:stoch_process}. For this specific setting, we say that a channel $\chn \in \mathfrak{E}_n$ is
able to $\epsilon$-thermalize the region $R$ if,
\begin{equation}
\label{eq:eps_therm_specific}
\norm{\chn \left( \omega_R \otimes \tau_{\beta}(H_R)^{\otimes n - 1} \right) - \tau_{\beta}(H_R)^{\otimes n}}_1 \leq \epsilon,
\end{equation}
that is, if the output of the channel is close, in trace distance, to the thermal state of region and bath.
\par
The quantity we seek to bound is $n_\epsilon$, that is, the minimum number of subsystems needed
to $\epsilon$-thermalize the region $R$, when the global dynamics is produced by a master equation of the
form given in Eq.~\eqref{eq:stoch_process},
\begin{equation}
\label{eq:opt_n_epsilon}
n_\epsilon := \min \left\{ n \in \N \, | \, \exists \, \chn \in \mathfrak{E}_n
\, : \,
\norm{\chn \left( \omega_R \otimes \tau_{\beta}(H_R)^{\otimes n - 1} \right) - \tau_{\beta}(H_R)^{\otimes n}}_1 \leq \epsilon \right\}.
\end{equation}
It is worth noting that, for the current setting, the channel $\bar{\chn}_n$ defined in Eq.~\eqref{eq:channel_csl}
belongs to the class of maps $\mathfrak{E}_n$. Indeed, this channel can be obtained from a master equation
describing stochastic collisions that occur between the region $R$ and each subsystem of the bath $B$, where the
collision is described by a swap operator,
\begin{equation}
\label{eq:master_eq_for_swap}
\frac{\partial \, \rho_{RB}(t)}{\partial \, t}
=
\sum_{i<j}^n \frac{1}{\tau} \left( U^{(i,j)}_{\text{swap}} \, \rho_{RB}(t) \, U^{(i,j) \, \dagger}_{\text{swap}} - \rho_{RB}(t) \right).
\end{equation}
For simplicity, in the above equation the collision rate is the same for all subsystems; however, the steady-state, and therefore the resulting channel associated to it, does not depend on the specifics of these rates (since we consider the infinite-time limit). Using the result of Prop.~\ref{prop:solution_poisson_process}, it is easy to show that, for an initial state $\rho_{RB}(t=0) = \omega_R \otimes \tau_{\beta}(H_R)^{\otimes n - 1}$, the state at time $t$ is
\begin{equation}
\rho_{RB}(t) = e^{- n \frac{t}{\tau}} \, \omega_R \otimes \tau_{\beta}(H_R)^{\otimes n - 1} + \left( 1 - e^{- n \frac{t}{\tau}} \right) \rho_{RB}^{(n)},
\end{equation}
where $\rho_{RB}^{(n)} = \frac{1}{n} \sum_{m=1}^n \tau_{\beta}(H_R)^{\otimes m-1} \otimes \omega_R \otimes \tau_{\beta}(H_R)^{\otimes n-m}$ is the state in Eq.~\eqref{eq:csl_state}. Thus, we see that, under this stochastic collision model, the system approaches its steady state exponentially fast in the collision rate $\tau^{-1}$, and in the number of subsystems composing region and bath. Finally, notice that, since the Hamiltonian of each subsystem is the same, the swap operator trivially conserves the energy.
\par
We are now able, with the help of Lemma~\ref{lem:csl}, to derive an upper bound on the quantity $n_{\epsilon}$.
The upper bound is obtained by providing an explicit protocol able to $\epsilon$-thermalize the system, and by
computing the number of subsystems $n$ needed for it.
\begin{theorem} [Upper bound on $n_{\epsilon}$]\label{thm_supp:upper}
	For a given Hamiltonian $H_R$, inverse temperature $\beta$, and a constant $\epsilon > 0$, we have that
	\begin{align}
	n_\epsilon \leq \frac{1}{\epsilon^2} \, 2^{\Dmax{}\left( \omega_R \| \tau_{\beta}(H_R) \right)}.
	\end{align}
\end{theorem}
\begin{proof}
	Let us consider the action of the channel $\bar{\chn}_n$, defined in Eq.~\eqref{eq:channel_csl}, on the initial state
	of the global system (region and bath) $\omega_R \otimes \tau_{\beta}(H_R)^{\otimes n - 1}$. It is easy to show that
	the final state of this channel is given by
	\begin{equation}
	\bar{\chn}_n \left( \omega_R \otimes \tau_{\beta}(H_R)^{\otimes n - 1} \right)
	=
	\frac{1}{n} \sum_{m=1}^n
	\tau_{\beta}(H_R)^{\otimes m-1} \otimes \omega_R \otimes \tau_{\beta}(H_R)^{\otimes n-m},
	\end{equation}
	which takes the same form of the state in the convex split lemma, see Eq.~\eqref{eq:csl_state}. Then, it directly
	follows from Lemma~\ref{lem:csl} that
	\begin{equation}
	\norm{\bar{\chn}_n \left( \omega_R \otimes \tau_{\beta}(H_R)^{\otimes n - 1} \right)
		-
		\tau_{\beta}(H_R)^{\otimes n}}_1^2 \leq {\frac{2^{\Dmax{}\left( \omega_R \| \tau_{\beta}(H_R) \right)}}{n}}.
	\end{equation}
	For the above trace distance to be smaller than $\epsilon$, that is, for the region to $\epsilon$-thermalize,
	we need a number of subsystems 
	\begin{equation}
	n = \frac{1}{\epsilon^2} 2^{\Dmax{}\left( \omega_R \| \tau_{\beta}(H_R) \right)},
	\end{equation}
	which closes the proof.
\end{proof}
\par
Let us now derive a lower bound to the quantity $n_{\epsilon}$. To do so, we first need to introduce two lemmata;
the first one is just a slight modification of Ref.~\cite[Fact~4 in Supp.~Mat.]{anshu_quantifying_2018}, where
we replace the quantum fidelity with the trace distance.
\begin{lemma}[Trace distance bound]
	\label{lem:cq_trace_dist}
	Consider two quantum states $\rho_A, \sigma_A \in \SH{\mcH_A}$, and let $\rho_{AB} \in \SH{\mcH_A \otimes \mcH_B}$
	be a classical-quantum state such that $\rho_A = \Trp{B}{\rho_{AB}}$. Then, there exists a classical-quantum state
	$\sigma_{AB} \in \SH{\mcH_A \otimes \mcH_B}$ such that $\sigma_A = \Trp{B}{\sigma_{AB}}$ and
	\begin{equation}
	\norm{\rho_{AB} - \sigma_{AB}}_1 \leq 2 \, \norm{\rho_A - \sigma_A}_1^{\frac{1}{2}}.
	\end{equation}
	Furthermore, $\supp ( \sigma_B ) \subseteq \supp( \rho_B )$.
\end{lemma}
\begin{proof}
	Under the hypotheses of this lemma, it was shown in Ref.~\cite[Fact~4 in Supp.~Mat.]{anshu_quantifying_2018} that
	\begin{equation}
	\label{eq:fid_equality}
	F(\rho_{AB},\sigma_{AB}) = F(\rho_A,\sigma_A) ,
	\end{equation}
	where $F(\rho,\sigma) = \norm{\sqrt{\rho} \sqrt{\sigma}}_1$ if the quantum fidelity between $\rho$ and $\sigma$.
	It is known that the trace distance between two states is linked to the quantum fidelity by the following chain of inequalities,
	\begin{equation}
	\label{eq:trace_dist_fid}
	1 - F(\rho,\sigma) \leq \frac{1}{2} \norm{\rho - \sigma}_1 \leq \sqrt{1 - F(\rho,\sigma)^2}.
	\end{equation}
	Therefore, we have that
	\begin{align}
	\norm{\rho_{AB} - \sigma_{AB}}_1^2
	\leq
	4 \left( 1 - F(\rho_{AB}, \sigma_{AB})^2 \right)
	=
	4 \left( 1 - F(\rho_A, \sigma_A)^2 \right)
	\leq
	4 \norm{\rho_A - \sigma_A}_1
	\end{align}
	where the first inequality follows from the rhs of Eq.~\eqref{eq:trace_dist_fid}, the equality from
	Eq.~\eqref{eq:fid_equality}, and the second inequality follows from the lhs of Eq.~\eqref{eq:trace_dist_fid}.
\end{proof}
We now recall and prove another result used in Ref.~\cite[Fact~6 in Supp.~Mat.]{anshu_quantifying_2018} that
we use to lower bound the quantity $n_{\epsilon}$ in the next theorem.
\begin{lemma}[\cite{anshu_quantifying_2018}]
	\label{lem:op_ineq_cq}
	Consider a classical-quantum state $\rho_{AB} \in \SH{\mcH_A \otimes \mcH_B}$, where $B$ is the classical part.
	Let $\Pi_B \in \BH{\mcH_B}$ be the projector onto the support of $\rho_B = \Trp{A}{\rho_{AB}}$. Then the
	following operator inequality holds,
	\begin{equation}
	\rho_{AB} \leq \rho_A \otimes \Pi_B.
	\end{equation}
\end{lemma}
\begin{proof}
	Since $\rho_{AB}$ is a classical-quantum state, there exists a probability distribution $\{ p_i \}_{i=1}^d$,
	being $d$ the dimension of the support of the classical part, and a set of states $\{ \rho_A^{(i)} \}_{i=1}^d$
	in $\SH{\mcH_A}$ such that $\rho_{AB} = \sum_{i=1}^d p_i \, \rho_A^{(i)} \otimes \ket{i}\bra{i}_B$. The reduced state
	on the quantum part of the system is $\rho_A = \sum_{i=1}^d p_i \, \rho_A^{(i)}$, and consequently we have that
	$\rho_A \otimes \Pi_B = \sum_{i,j = 1}^d p_i \, \rho_A^{(i)} \otimes \ket{j}\bra{j}_B$. Then, it is easy to show that
	the operator
	\begin{equation}
	\rho_A \otimes \Pi_B - \rho_{AB}  = \sum_{i \neq j}^d p_i \, \rho_A^{(i)} \otimes \ket{j}\bra{j}_B,
	\end{equation}
	is positive semi-definite, since it is composed by a positive mixture of states.
\end{proof}
We are now in the position to derive a lower bound for the quantity $n_{\epsilon}$ defined in Eq.~\eqref{eq:opt_n_epsilon}.
Our proof is inspired by the one used in Ref.~\cite[Sec.~3.2 in Supp.~Mat.]{anshu_quantifying_2018} to derive a converse
to the convex split lemma.
\begin{theorem}[Lower bound on $n_{\epsilon}$]
	\label{thm:lower}
	For a given $\beta$ and $\epsilon > 0$, and a Hamiltonian $H_R$ satisfying the energy subspace condition
	(see Def.~\ref{def:ESC}), we have
	\begin{equation}
	n_\epsilon \geq 2^{\Dmax{2 \sqrt{\epsilon}+\delta}\left( \omega_R \| \tau_{\beta}(H_R) \right)},
	\end{equation}
	where $\delta = \norm{\Delta(\omega_R) - \omega_R}_1$ quantifies the distance from the state of the region $\omega_R$
	and its decohered version $\Delta(\omega_R)$.
\end{theorem}
\begin{proof}
	For the sake of simplicity, in the following we refer to the initial state as $\rho_{RB} = \omega_R \otimes
	\tau_{\beta}(H_R)^{\otimes n_{\epsilon} - 1}$, and to the target state as $\tau_{RB} = \tau_{\beta}(H_R)^{
		\otimes n_{\epsilon}}$. For a fixed parameter $\epsilon > 0$, let $\hat{\chn} \in \mathfrak{E}_{n_{\epsilon}}$
	be the (not necessarily unique) channel such that
	\begin{equation}
	\norm{ \hat{\chn} \left( \rho_{RB} \right) - \tau_{RB} }_1
	\leq \epsilon.
	\end{equation}
	Let us now introduce the channel $\Delta$, that decoheres the system in the energy eigenbasis of the total Hamiltonian $H = \sum_{i=1}^{n_{\epsilon}} H_R^{(i)}$. In the proof of Prop.~\ref{prop:csl_semi-optimality}, we show that the action of this channel commutes with that of any channel in $\mathfrak{E}_{n_{\epsilon}}$. Thus, using monotonicity of the trace distance under CPTP maps, we have that
	\begin{equation}
	\norm{\hat{\chn} \circ \Delta \left( \rho_{RB} \right) - \tau_{RB} }_1
	=
	\norm{\Delta \circ \hat{\chn} \left( \rho_{RB} \right) - \Delta ( \tau_{RB} ) }_1
	\leq
	\norm{ \hat{\chn} \left( \rho_{RB} \right) - \tau_{RB} }_1
	\leq \epsilon,
	\end{equation}
	where we additionally used the fact that $\tau_{RB}$ is diagonal in the energy eigenbasis. In Prop.~\ref{prop:csl_optimality} we show that, when the Hamiltonian $H_R$ satisfies the ESC and the input and target states are diagonal, the optimal thermalization is achieved via the channel $\bar{\chn}_{n_{\epsilon}}$ of Eq.~\eqref{eq:channel_csl}. Thus, it holds that
	\begin{equation}
	\norm{\bar{\chn}_{n_{\epsilon}} \circ \Delta \left( \rho_{RB} \right) - \tau_{RB} }_1
	\leq
	\epsilon
	\end{equation}
	We now make use of the above  bound on the trace distance, and of the specific form of the channel
	$\bar{\chn}_{n_{\epsilon}}$ to derive a lower bound on the quantity $n_{\epsilon}$. Let us first recall
	that the channel 
	\begin{equation}
	\bar{\chn}_{n_{\epsilon}}(\cdot) = \sum_{i=1}^{n_{\epsilon}} \frac{1}{n_{\epsilon}} \,
	U^{(1,i)}_{\text{swap}} \, \cdot \, U^{(1,i) \, \dagger}_{\text{swap}},
	\end{equation}
	where the unitary operator
	$U^{(1,i)}_{\text{swap}} \in \BH{\mcH_R \otimes \mcH_B}$ swaps the state of the first subsystem (the
	region $R$) with that of the $i$-th subsystem (belonging to the bath $B$). We can dilate this map by
	introducing the following unitary operation,
	\begin{equation}
	V_{RBA} = \sum_{i=1}^{n_{\epsilon}} \, U^{(1,i)}_{\text{swap}} \otimes \ket{i}\bra{i}_A,
	\end{equation}
	which acts over the region, the bath, and an ancillary system $A$ of dimension $n_{\epsilon}$. Then,
	for an ancillary system described by the state $\rho_A = \sum_{i=1}^{n_{\epsilon}} {n_{\epsilon}}^{-1}
	\, \ket{i}\bra{i}_A$, we can define the global state $\tilde{\rho}_{RBA} = V_{RBA}
	\left( \Delta(\rho_{RB}) \otimes \rho_A \right) V_{RBA}^{\dagger}$. This is a classical-quantum state,
	and when the ancillary subsystem $A$ is traced out it coincides with $\bar{\chn}_{n_{\epsilon}} \circ \Delta \left( \rho_{RB} \right)$.
	\par
	From Lemma~\ref{lem:cq_trace_dist} it follows that there exists a classical-quantum extension of the target state
	$\tau_{RB}$, which we refer to as $\tau_{RBA}$ (where $A$ is the classical part of the state with dimension $n_{\epsilon}$),
	such that $\norm{\tilde{\rho}_{RBA} - \tau_{RBA}}_1 \leq 2 \, \sqrt{\epsilon}$. Furthermore, since $\tau_{RBA}$ is
	classical-quantum, we have that the operator inequality $\tau_{RBA} \leq \tau_{RB} \otimes \id_A$ holds, see
	Lemma.~\ref{lem:op_ineq_cq}. Using this operator inequality, the fact that $\rho_A = {\id_A}/{n_{\epsilon}}$,
	and the definition of the max-relative entropy it follows that $\Dmax{}(\tau_{RBA}|\tau_{RB} \otimes \rho_A)
	\leq \log n_{\epsilon}$. By monotonicity of this measure with respect to CPTP maps, we have that
	\begin{equation}
	\label{eq:max_rel_as_lower_bound}
	\Dmax{}\left( \Trp{BA}{ V_{RBA}^{\dagger} \, \tau_{RBA} \, V_{RBA}} \, \right| \,
	\Trp{BA}{ V_{RBA}^{\dagger} \left( \tau_{RB} \otimes \rho_A \right) V_{RBA}}
	\Big) \leq \log n_{\epsilon}.
	\end{equation}
	Let us consider the first argument of the above max-relative entropy. Using the triangle inequality, we can map the problem to the decohered case,
	\begin{align*}
	\norm{ \Trp{BA}{ V_{RBA}^{\dagger} \, \tau_{RBA} \, V_{RBA}} - \omega_R }_1
	&=
	\norm{ \Trp{BA}{ V_{RBA}^{\dagger} \, \tau_{RBA} \, V_{RBA}} - \Delta(\omega_R) + \Delta(\omega_R) - \omega_R }_1 \\
	&\leq
	\norm{ \Trp{BA}{ V_{RBA}^{\dagger} \, \tau_{RBA} \, V_{RBA}} - \Delta(\omega_R)}_1
	+
	\norm{ \Delta(\omega_R) - \omega_R }_1.
	\end{align*}
	The first term of the above sum can be further simplified,
	\begin{align*}
	\norm{ \Trp{BA}{ V_{RBA}^{\dagger} \, \tau_{RBA} \, V_{RBA}} - \Delta(\omega_R)}_1
	&\leq
	\norm{ V_{RBA}^{\dagger} \, \tau_{RBA} \, V_{RBA} -  \Delta(\rho_{RB}) \otimes \rho_A }_1 \\
	&=
	\norm{ \tau_{RBA} - V_{RBA} \left( \Delta(\rho_{RB}) \otimes \rho_A \right) V_{RBA}^{\dagger} }_1 \\
	&=
	\norm{ \tau_{RBA} - \tilde{\rho}_{RBA}} \leq 2 \, \sqrt{\epsilon},
	\end{align*}
	where in the first inequality we use the monotonicity of the trace distance under partial trace, and the fact that $\Delta(\rho_{RB}) = \Delta(\omega_R) \otimes \tau_{\beta}(H_R)^{\otimes n_{\epsilon} - 1}$ since $\tau_{\beta}(H_R)$ has no coherence in the energy eigenbasis of $H_R$. The first equality follows from the unitary invariance of the trace distance, and the last inequality follows from how we have defined $\tau_{RBA}$. Thus, the initial state of the region $\omega_R$ is within a ball of radius $2 \, \sqrt{\epsilon} + \delta$ from the state in the first argument of the max-relative entropy in Eq.~\eqref{eq:max_rel_as_lower_bound}, where $\delta = \norm{ \Delta(\omega_R) - \omega_R }_1$.
	\par
	The state in the second argument of the max-relative entropy can instead be explicitly computed,
	\begin{equation*}
	\Trp{BA}{ V_{RBA}^{\dagger} \left( \tau_{RB} \otimes \rho_A \right) V_{RBA}}
	=
	\Trp{B}{ \bar{\chn}_{n_{\epsilon}} \left( \tau_{RB} \right) }
	=
	\Trp{B}{ \tau_{RB} } = \tau_{\beta}(H_R),
	\end{equation*}
	where the first equality follows from the definition of $V_{RBA}$, while the second one from the fact
	that $\tau_{RB} = \tau_{\beta}(H_R)^{\otimes n_{\epsilon}}$ is invariant under permutation. As a
	result, we can replace Eq.~\eqref{eq:max_rel_as_lower_bound} with the following one,
	\begin{equation}
	\Dmax{2 \, \sqrt{\epsilon} + \delta}\left( \omega_R \, | \, \tau_{\beta}(H_R) \right) \leq \log n_{\varepsilon},
	\end{equation}
	which closes the proof.
\end{proof}
\section{Supplementary Note 3: Discussion on choice of bath used in model}
Given the model studied in this work, a question arises whether the choice of bath is a suitable one, since this directly relates to the meaningfulness of the lower and upper bounds on bath size, which are derived in this work. It is worthwhile to note that, in the scientific literature that addresses thermalization in many-body systems (hence in particular MBL systems), mostly long-time limits of master equations are considered~\cite{fischer2016dynamics}. Such a setting would correspond to a bath of infinite size and no memory. On the other hand, other works on MBL thermalization use a specific finite bath~\cite{luitz_how_2017,Goihl19} which is modeled as part of the chain, with regions that have low disorder. These works so far feature only thermalization with respect to local observables, which are a much more lenient measure of thermalization and do not fully capture the non-local, non-thermal aspects of the system. 

From the resource-theoretic point of view, not much has been said so far about the required bath sizes for arbitrary Hamiltonians and state transitions~\cite{scharlau2018quantum}. In all resource theoretic settings to date, the final state of the bath is relatively unimportant, since it is always discarded, and the bath acts simply as a heat source used to thermalize systems. However, discarding the bath is not a suitable consideration in the current setting, as we mention in the main text, since full thermalization of the system can always be achieved with a bath composed by one single copy of the system in a thermal state. Clearly the resource-theoretic setting is still useful in general, since it does not simply characterize transformations achieving full-thermalization, but it studies state transitions that allow work extraction (usually modeled as a transition between two non-equilibrium, non-thermal states).

A question of concern is what the minimal bath structure necessary to thermalize a system is, given globally energy-preserving interactions. The populations of each energy level on the system need to be altered, which means that the bath must contain energy gaps that are present in the system. We mention in the main text two reasons in choosing the particular bath, namely tractability of the problem and its relevance in some experiments. However, one also observes that the level structure of our bath is such that it contains all features of the system Hamiltonian, and no additional/unnecessary features (such as energy gaps that are not present in the system), which gives good reason to think that our bath should not be unnecessarily large. 

Nevertheless, one can of course also consider other bath models. An alternative, conceivable bath example to use are qubit baths, namely baths consisting of a collection of qubits where the energy gaps correspond to gaps of the system. A second example would be a collection of bosonic modes with frequencies corresponding to each energy gap present in the system~\cite{lostaglio2018elementary}; however, the dimension of such a bath will be infinite to begin with, and would automatically satisfy the lower bound derived in our work. The number of different required frequencies, however, might be an interesting research question for future work, should such bath models be of particular interest. Both of the above example suffer a particular disadvantage; since energy eigenstates of many-body systems are generally very non-local, this would mean that the operations required to thermalize system+bath, according to energy-preserving operations, would also be highly non-local and thus convoluted.

\section{Supplementary Note 4: Optimality of the stochastic swapping collision model}
In this section we show that the stochastic collision model introduced in Eq.~\eqref{eq:channel_csl} allows us, under
some assumptions on the system’s Hamiltonian and initial state, to obtain the optimal thermalization for a given number of subsystems
$n \in \N$. Specifically, we show that within the class of energy-preserving random unitary channels acting on $n$
subsystem, the map $\bar{\chn}_n$ provides the minimum value of $\epsilon$ in Eq.~\eqref{eq:eps_therm_specific},
when the state $\omega_R$ is diagonal in the energy eigenbasis, Prop.~\ref{prop:csl_optimality}. Additionally, we are able to bound the performance of the swapping collision model in the situation in which the state $\omega_R$ has coherence in the energy eigenbasis, Prop.~\ref{prop:csl_semi-optimality}, and we show that this channel is able to efficiently thermalize the system when the state has low coherence.
In order to prove the above statement, we need to introduce the following lemmata. The first lemma concerns the power
of energy-preserving random unitary channels in modifying the spectrum of a quantum state.
\begin{lemma}[Power of energy-preserving random unitary channels in modifying the spectrum of a quantum state]
	\label{lem:preserve_weights}
	Consider a Hilbert space $\mcH$, a Hamiltonian $H = \sum_E E \, \Pi_E$ where $\left\{ \Pi_E \right\}_E$ is the set of projectors
	onto the energy subspaces, and a state $\rho \in \SH{\mcH}$. Given any channel of the form $\mathcal{E}(\cdot) = \sum_k p_k \, U_k \,
	\cdot \, U_k^{\dagger}$, where $\left\{ p_k \right\}_k$ is a probability distribution and $\left\{ U_k \right\}_k$ is a set of energy
	preserving unitaries $\left[ U_k , H \right] = 0$, we have that
	\begin{equation}
	\tr \left[ \mathcal{E}(\rho) \, \Pi_E \right] = \tr \left[  \rho \, \Pi_E \right] \quad \forall \, E.
	\end{equation}
\end{lemma}
\begin{proof}
	Due to the fact that each unitary $U_k$ commutes with the Hamiltonian $H$, we have that $U_k^{\dagger} \, \Pi_E \, U_k = \Pi_E$
	for all $k$ and $E$. Then,
	\begin{align*}
	\tr \left[ \mathcal{E}(\rho) \, \Pi_E \right] &=
	\sum_k p_k \tr \left[ U_k \, \rho \, U_k^{\dagger} \, \Pi_E \right]
	= \sum_k p_k \tr \left[ \rho \, U_k^{\dagger} \, \Pi_E \, U_k  \right]
	= \sum_k p_k \tr \left[ \rho \, \Pi_E \right] = \tr \left[ \rho \, \Pi_E \right].
	\end{align*}
\end{proof}
In the next lemma we consider a family of quantum states with fixed weights in different subspaces, and we explicitly
construct a state in this family which minimizes the distance to a given state outside the family.
\begin{lemma}[Fixed weights subspaces]
	\label{lem:opt_state}
	Consider a Hilbert space $\mcH$ and a set of orthogonal projectors $\{\Pi_i\}_i$ on $\mcH$ such that $\sum_i \Pi_i =
	\id$. Given a set of probabilities $\{p_i\}_i$, let $S = \left\{ \rho \in \SH{\mcH} \ | \ \tr \left[\rho \, \Pi_i \right] = p_i \ \forall \, i  \right\}$.
	Furthermore, assume that a state $\sigma \in \mc S(\mcH)$ has the form $ \sigma = \sum_i q_i \sigma_i $, where each
	$ \Pi_i\sigma_i\Pi_i = \sigma_i  $ and $ \tr(\sigma_i) =1 $. Then, the state $ \bar{\rho} = \sum_i p_i \sigma_i$ minimizes
	the trace distance to $\sigma$ over $S$, that is,
	\begin{align}
	\bar{\rho} \in \displaystyle\argmin_{\rho \in S} \norm{\rho - \sigma}_1.
	\end{align}
\end{lemma}
\begin{proof}
	In order to find the optimal state in the family $S$, we introduce the following CPTP maps which describes a quantum
	instrument, $\mathcal{E}( \cdot ) = \sum_i \tr \left[ \, \cdot \, \Pi_i \right] \, \sigma_i$. It is easy to see that $\sigma$ is
	left invariant by the above map, and that for any $\rho \in S$, $\mathcal{E}( \rho ) = \sum_i p_i \, \sigma_i$. Using the
	monotonicity of the trace distance under CPTP maps, we find that for any $\rho \in S$,
	\begin{equation*}
	\left\| \rho - \sigma \right\|_1 \geq
	\left\| \mathcal{E}( \rho) - \mathcal{E}( \sigma) \right\|_1 =
	\left\| \sum_i p_i\, \sigma_i -  \sigma \right\|_1.
	\end{equation*}
\end{proof}
The next lemma we prove require the system Hamiltonian to satisfy the following condition,
\begin{definition}[Energy subspace condition]
	\label{def:ESC}
	Given a Hamiltonian $H$, we say that it fulfills the ESC iff for any $n \in \N$, given the set of energy levels
	$\left\{E_k \right\}_{k=1}^d$ of the Hamiltonian $H$, we have that for any vectors $m, m' \in \N^d $ with
	the same normalization factor, namely $\sum_k m_k = \sum_k m'_k = n $,
	\begin{equation}
	\label{item_supp:opt_cond}
	\sum_k m_k E_k \neq \sum_k m'_k E_k .
	\end{equation}
\end{definition}
The following lemma concern the state $\rho^{(n)}$ introduced in Eq.~\eqref{eq:csl_state}. This is the central object
of the convex split lemma, as well as of our Prop.~\ref{prop:csl_optimality}. We show that, under the above assumption
on the Hamiltonian of the system, this state is uniformly distributed over each energy subspace.
\begin{lemma}[Uniform distribution over each energy subspace]
	\label{lem:csl_uniform}
	Consider a Hilbert space $\mcH$ of dimension $d$, a Hamiltonian $H^{(1)}$, and two states $\rho, \sigma \in \SH{\mcH}$.
	For any $n \in \N$, consider the state $\rho^{(n)} \in \SH{\mcH^{\otimes n}}$ defined in Eq.~\eqref{eq:csl_state}. This state
	is uniformly mixed over each energy subspace of the total Hamiltonian $H^{(n)} = \sum_{i=1}^n H^{(1)}_i$ if
	\begin{enumerate}
		\item $ H^{(1)} $ satisfies the ESC condition, and
		\item \label{item:states} The states $\rho, \sigma$ are diagonal in the eigenbasis of $H^{(1)}$.
	\end{enumerate}
\end{lemma}
\begin{proof}
	Let the eigenbasis of $H^{(1)}$ be $\left\{ \ket{i} \right\}_{i=1}^d$, and consider the basis of $\mcH^{\otimes{n}}$
	given by
	\begin{equation}\label{eq:totaleigenbasis}
	\left\{\ket{\mathbf{i}} = \ket{i_1} \otimes \ldots \otimes \ket{i_n} \right\}_{i_1, \ldots, i_n = 1}^d,
	\end{equation}
	where $\mathbf{i}^T = (i_1, \dots, i_n)$. It is easy to see that the set of vectors in Eq.~\eqref{eq:totaleigenbasis}
	form an eigenbasis of the total Hamiltonian $H^{(n)}$. Let us now introduce a way of partitioning the vectors
	$ \mathbf{i} $, namely by characterizing them w.~r.~t.~the number of elements in $ \mathbf{i} $ corresponding to
	each distinct energy eigenvalue on $ H^{(1)} $. For a simple example when $ H^{(1)} $ is a qubit, then this scheme
	characterizes each basis vector $ |\mathbf{i}\rangle $ by the Hamming weight of $ \mathbf{i} $. Given a tuple
	\begin{equation}
	\label{eq:def_tuple_type}
	\mathbf{n} := \left( n_1, n_2 , \ldots, n_d \right),
	\end{equation}
	let $I_{\mathbf{n}}$ denote the set of vectors $\mathbf{i}$ such that if $ \mathbf{i}\in I_{\mathbf{n}} $, then
	$ \mathbf{i} $ contains $ n_i $ elements that are equal to $ i $. Furthermore, let $\Pi_{\mathbf{n}} =
	\sum_{\mathbf{i} \in I_{\mathbf{n}}} \proj{i}$ denote the projector onto this subspace.
	\par
	Note that basis vectors $ |\mathbf{i}\rangle,|\mathbf{i'}\rangle $, where $ \mathbf{i},\mathbf{i'}\in I_{\mathbf{n}}$
	correspond to the same tuple ${\mathbf{n}}$, have the same energy. On the other hand, the ESC guarantees that if $ \mathbf{i},\mathbf{i'} $ are not in the same set $I_{\mathbf{n}}$, then they belong to different
	energy subspaces. In other words, the set of $ \lbrace \Pi_\mathbf{n}\rbrace $ coincides with the set of energy
	subspaces of the total Hamiltonian $H^{(n)}$. We now want to show that $\Pi_\mathbf{n} \rho^{(n)} \Pi_\mathbf{n}
	\propto \id$, namely that $ \rho^{(n)} $ is uniform in a fixed energy subspace. To do so, we may calculate the overlap
	$ \bra{\mathbf{i}} \rho^{(n)}\ket{\bf i} $ and show that it does not depend on the particular basis vector $ \ket{\bf i} $,
	but only the tuple $ {\bf n} $ such that $ {\bf i} \in I_{\bf n} $.
	\par
	Given a particular $ \ket{\bf i} $ corresponding to tuple $ {\bf n} $, notice that for each $ i = 1,\dots, d $, each
	single-system state $\ket{i}$ describes a total $n_i$ number of the $n$ subsystems; and we denote the indices
	of these subsystems as $\{ j^{(i)}_{\ell} \}_{\ell=1}^{n_i}$.  Let us first observe that for all $ i=1,\dots,d $
	and $ \ell = 1,\dots,n_i $, the following overlap holds,
	\begin{equation}
	\label{eq:overlap_elem_mixture}
	\bra{\bf i} \sigma_1 \otimes \ldots \otimes \rho_{j^{(i)}_{\ell}} \otimes \ldots \otimes \sigma_n \ket{\bf i}
	=
	p_i \, Q_i,
	\end{equation}
	where $p_i = \bra{i} \rho \ket{i}$ and $Q_i = q_i^{n_i - 1} \prod_{j \neq i}^d q_j^{n_j}$, with $q_j = \bra{j} \sigma \ket{j}$.
	Here, one can already observe that the particular location, characterized by $ \ell $, does not affect the overlap;
	Eq.~\eqref{eq:overlap_elem_mixture} is fully characterized by $ i $. We now proceed to rewrite the state $ \rho^{(n)} $
	can be rewritten as
	\begin{align}\label{eq:tau_n}
	\rho^{(n)} &= \frac{1}{n} \sum_{j=1}^n \sigma_1\otimes\cdots\otimes\rho_j\otimes\sigma_n
	=\frac{1}{n} \sum_{i=1}^d\sum_{\ell =1}^{n_i} \sigma_1 \otimes \ldots \otimes \rho_{j^{(i)}_{\ell}} \otimes \ldots \otimes \sigma_n.
	\end{align}
	Using Eq.~\eqref{eq:overlap_elem_mixture} and \eqref{eq:tau_n} together, we find that  \begin{align*}
	\bra{\bf i} \rho^{(n)} \ket{\bf i} 
	&= \sum_{i=1}^d \sum_{\ell=1}^{n_i} \frac{1}{n} \,
	\bra{\bf i} \sigma_1 \otimes \ldots \otimes \rho_{j^{(i)}_{\ell}} \otimes \ldots \otimes \sigma_n \ket{\bf i}
	= \sum_{i=1}^d \frac{n_i}{n} p_i \, Q_i.
	\end{align*}
	Note that $ \lbrace p_i\rbrace $ and $ \lbrace Q_i\rbrace $ are fully determined by $ \rho $, $ \sigma $, $ H^{(1)} $ and
	$ n $, which are fixed from the beginning. Therefore, for all ${\bf i} \in I_{\bf n}$, the overlap $ \bra{\bf i}\rho^{(n)}\ket{\bf i} $
	depends only on $ \bf n $, which concludes the proof.
\end{proof}
We are now able to prove the optimality of the stochastic collision model $\bar{\chn}_n$, within the class of channels
$\mathfrak{E}_n$, for the thermalization of a system in contact with an external bath.
\begin{proposition}[Optimality of the stochastic collision model for diagonal states]
	\label{prop:csl_optimality}
	Given $ n \in\N$, let $H^{(1)}$ be a Hamiltonian satisfying the ESC as defined in Def.~\ref{def:ESC}. Then, for any two states $\rho,
	\sigma \in \SH{\mcH}$ diagonal in the energy eigenbasis, the channel $\bar{\chn}_n$ of Eq.~\eqref{eq:channel_csl}
	allows to achieve the optimal thermalization, that is
	\begin{align}
	\bar{\chn}_n \in \argmin_{\mathcal{E} \in \mathfrak{E}_n}
	\norm{\mathcal{E}(\rho \otimes \sigma^{\otimes n-1})- \sigma^{\otimes n}}_1.
	\end{align}
\end{proposition}
\begin{proof}
	Let us first notice that Cor.~\ref{cor:E_sub_R} tells us that the family of channels $\mathfrak{E}_n$ is a subset of the class of
	energy-preserving random unitary channels. In Lem.~\ref{lem:preserve_weights}, we have shown that this latter class of channels
	cannot modify the weight associated with each energy subspace in the initial state $\rho \otimes \sigma^{\otimes n-1}$. Thus,
	the channels in $\mathfrak{E}_n$ can only modify the form of the distribution in each of these subspaces. Furthermore, in
	Lem.~\ref{lem:opt_state} we showed that the state minimizing the trace distance to the target state $\sigma^{\otimes n}$ is the
	one with the same distribution over each energy subspace and with no coherence in the energy eigenbasis (since the target state is diagonal).
	It is easy to see that $\sigma^{\otimes n}$ has a uniform distribution over each energy subspace. Thus, a channel $\epsilon \in \mathfrak{E}_n$
	minimizes the trace distance $\norm{\mathcal{E}(\rho \otimes \sigma^{\otimes n-1})- \sigma^{\otimes n}}_1$ if its output
	state is uniform over the energy subspaces. But in Lem.~\ref{lem:csl_uniform} we showed that, if the Hamiltonian $H^{(1)}$
	satisfies ESC, then the state $\rho^{(n)} = \bar{\chn}_n(\rho \otimes \sigma^{\otimes n-1})$ is uniform over
	the energy subspaces , which concludes the proof.
\end{proof}
It is worth noting that the above proposition applies to diagonal states only. In the context of MBL systems, the thermal
state $\tau_{\beta}(H_R)$ is (by construction) diagonal in the energy eigenbasis, but the same needs not to apply for the
reduced state $\omega_R$ of the infinite-time average. For this reason, we introduce the following proposition, characterizing the limitations of the collisional model of Eq.~\eqref{eq:channel_csl} when the initial state of the region has coherence in the energy eigenbasis.
\begin{proposition}[Thermalization bound for the stochastic collision model]
	\label{prop:csl_semi-optimality}
	Given $n \in\N$, let $H^{(1)}$ be a Hamiltonian satisfying the ESC as defined in Def.~\ref{def:ESC}. Then, for any two states $\rho,
	\sigma \in \SH{\mcH}$, where $\sigma$ is diagonal in the energy eigenbasis, the following bound on the thermalization achievable
	via the channel $\bar{\chn}_n$ of Eq.~\eqref{eq:channel_csl} holds,
	\begin{align}
	\label{eq:semi-opt_bound}
	\norm{\bar{\chn}_n(\rho \otimes \sigma^{\otimes n-1})- \sigma^{\otimes n}}_1
	\leq
	\norm{\chn_{\mathrm{opt}}(\rho \otimes \sigma^{\otimes n-1})- \sigma^{\otimes n}}_1
	+
	\norm{\rho - \Delta(\rho)}_1,
	\end{align}
	where $\chn_{\mathrm{opt}} \in \mathfrak{E}_n$ is the channel achieving the optimal thermalization, and $\Delta(\rho)$ is the
	decohered version of the state $\rho$.
\end{proposition}
\begin{proof}
	As first step, let us notice that the decohering channel $\Delta_n$, which removes coherence in the energy eigenbasis, is defined as $\Delta_n(\cdot) = \sum_E \Pi_E \, \cdot \, \Pi_E$, where $\Pi_E$ is the eigenprojector of the total Hamiltonian $H = \sum_{i=1}^n H^{(1)}_i$ associated with the energy $E$. Since the set of channels we are optimizing over is a subset of energy-preserving random unitary channels, see Cor.~\ref{cor:E_sub_R}, it is easy to show that the action of $\Delta_n$ commutes with that of any channel $\chn \in \mathfrak{E}_n$,
	\begin{align}
	\label{eq:comm_relation_decoh}
	\Delta \circ \chn
	&=
	\sum_E \Pi_E \, \left( \sum_k p_k \, U_k \, \cdot \, U_k^{\dagger} \right) \, \Pi_E
	=
	\sum_{E,k} p_k \left( \Pi_E U_k \right) \, \cdot \, \left( \Pi_E U_k \right)^{\dagger}
	=
	\sum_{E,k} p_k \left( U_k \Pi_E \right) \, \cdot \, \left( U_k \Pi_E \right)^{\dagger} \nonumber \\
	&=
	\sum_k p_k \, U_k \left( \sum_E \Pi_E \, \cdot \, \Pi_E \right) U_k^{\dagger}
	=
	\chn \circ \Delta,
	\end{align}
	where we used $\chn(\cdot) = \sum_k p_k \, U_k \, \cdot \, U_k^{\dagger}$, and the third equality follows from the fact that $[U_k,H] = 0$ for all $k$. Furthermore, due to the form of the global Hamiltonian $H$ we have that, when $\sigma$ is diagonal in the energy eigenbasis, $\Delta_n(\rho \otimes \sigma^{\otimes n-1}) = \Delta_1(\rho) \otimes \sigma^{\otimes n-1}$, where $\Delta_1$ decoheres with respect to the energy eigenbasis of $H^{(1)}$. For simplicity, in the following we suppress the subscript from the map $\Delta_n$, since the number of subsystems the map is action over should be clear from its argument.
	\par
	Let us now consider the trace distance between the output of the channel $\bar{\chn}_n$ and the target state,
	\begin{align}
	\label{eq:out_coh_bound}
	\norm{\bar{\chn}_n(\rho \otimes \sigma^{\otimes n-1}) - \sigma^{\otimes n}}_1
	&=
	\norm{\bar{\chn}_n(\rho \otimes \sigma^{\otimes n-1}) - \Delta \circ \bar{\chn}_n(\rho \otimes \sigma^{\otimes n-1})
		+ \Delta \circ \bar{\chn}_n(\rho \otimes \sigma^{\otimes n-1}) - \sigma^{\otimes n}}_1 \nonumber \\
	&\leq
	\norm{\bar{\chn}_n(\rho \otimes \sigma^{\otimes n-1}) - \Delta \circ \bar{\chn}_n(\rho \otimes \sigma^{\otimes n-1})}_1
	+
	\norm{\Delta \circ \bar{\chn}_n(\rho \otimes \sigma^{\otimes n-1}) - \sigma^{\otimes n}}_1,
	\end{align}
	where we have used triangle inequality. Let us consider the first term of the last line above,
	\begin{align}
	\label{eq:coh_dist_bound}
	\norm{\bar{\chn}_n(\rho \otimes \sigma^{\otimes n-1}) - \Delta \circ \bar{\chn}_n(\rho \otimes \sigma^{\otimes n-1})}_1
	&=
	\norm{\bar{\chn}_n(\rho \otimes \sigma^{\otimes n-1}) - \bar{\chn}_n(\Delta(\rho) \otimes \sigma^{\otimes n-1})}_1
	\leq
	\norm{\rho \otimes \sigma^{\otimes n-1} - \Delta(\rho) \otimes \sigma^{\otimes n-1}}_1 \nonumber \\
	&=
	\norm{\rho - \Delta(\rho)}_1
	\end{align}
	where the first equality follows from Eq.~\eqref{eq:comm_relation_decoh}, the inequality from the monotonicity of the
	trace distance with respect to CPTP maps, and the last equality from the fact that $\sigma$ is a normalized state.
	The second term of Eq.~\eqref{eq:out_coh_bound} can instead be bounded as follows,
	\begin{align}
	\label{eq:opt_bound_coh}
	\norm{\Delta \circ \bar{\chn}_n(\rho \otimes \sigma^{\otimes n-1}) - \sigma^{\otimes n}}_1
	&=
	\norm{\bar{\chn}_n(\Delta(\rho) \otimes \sigma^{\otimes n-1}) - \sigma^{\otimes n}}_1
	\leq
	\norm{\chn_{\mathrm{opt}}(\Delta(\rho) \otimes \sigma^{\otimes n-1}) - \sigma^{\otimes n}}_1 \nonumber \\
	&=
	\norm{\Delta \circ \chn_{\mathrm{opt}}(\rho \otimes \sigma^{\otimes n-1}) - \Delta(\sigma^{\otimes n})}_1
	\leq
	\norm{\chn_{\mathrm{opt}}(\rho \otimes \sigma^{\otimes n-1}) - \sigma^{\otimes n}}_1
	\end{align}
	where, again, the first equality follows from the fact that the actions of $\Delta$ and $\bar{\chn}_n$ commute with
	each other. The first inequality follows from Prop~\eqref{prop:csl_optimality}, where we have shown that $\bar{\chn}_n$
	is the optimal channel when the input and target states are diagonal; the channel $\chn_{\mathrm{opt}} \in \mathfrak{E}_n$
	is instead the optimal channel when the input state $\rho$ is not dechoered. The second equality follows from
	Eq.~\eqref{eq:comm_relation_decoh} and the fact that $\sigma$ is diagonal. The final inequality follows from the
	monotonicity of the trace distance. Combining Eqs.~\eqref{eq:coh_dist_bound} and \eqref{eq:opt_bound_coh} into
	Eq.~\eqref{eq:out_coh_bound} proves the proposition.
\end{proof}
The above result provides information on the thermalization power of the channel $\bar{\chn}_n$ for states with coherence in the energy eigenbasis. For systems in the MBL phase, one expects almost all the eigenstates of the Hamiltonian to be close to product states~\cite{friesdorf_many-body_2015}, and hence $\omega_R$ should have relatively small and strongly decaying off-diagonal terms in the eigenbasis of $H_R$. To support this statement, we numerically compute the coherence in the energy eigenbasis contained in the reduced state $\omega_R$ for the disordered Heisenberg chain we studied in the main text, see Supplementary Figure~\ref{fig:coh_reduced}. As a result, the above proposition tells us that the channel $\bar{\chn}_n$ is able to efficiently thermalizing an MBL system, conditioned on the fact that the ESC condition is satisfied.
\begin{figure}[t]
	\centering
	\includegraphics[width=0.45\textwidth]{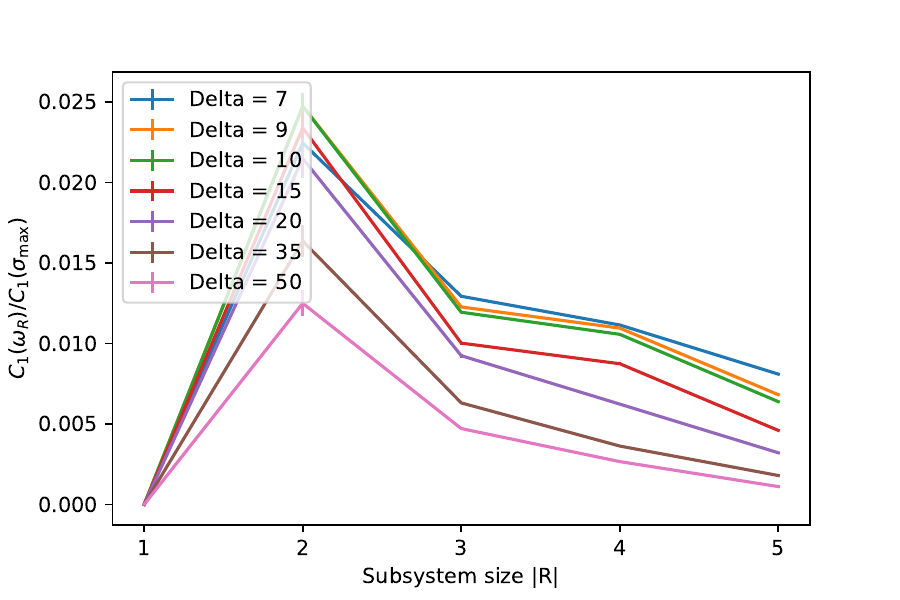}
	\includegraphics[width=0.45\textwidth]{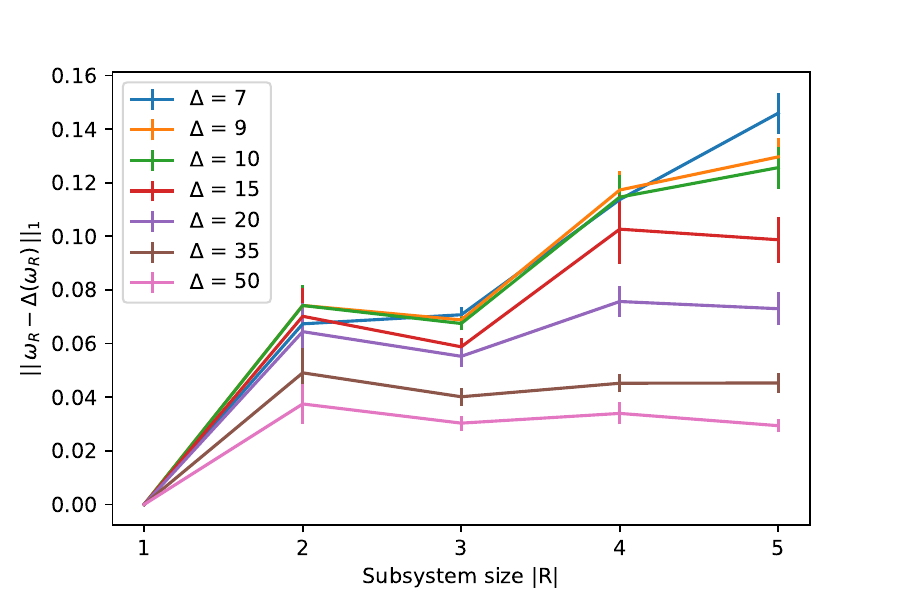}
	\caption{We consider the disordered Heisenberg chain with $L=12$ sites for different disordered strengths. We compute the coherence in the energy eigenbasis (for the operator $H_R$) of the state $\omega_R$, where $R$ is a region in the chain with variable size. {\bf Left.} We compute the $l_1$-norm of coherence, defined as $C_1(\rho) = \sum_{i \neq j} |\rho_{ij}|$, see for instance Ref.~\cite{streltsov_maximal_2018}, for the state $\omega_R$ as a function of the size of the region $|R|$. Since this measure depends on the dimension of the state, we normalize it by its maximum value achieved by the state $\sigma_{\text{max}} = \ket{\psi_{\text{max}}}\bra{\psi_{\text{max}}}$, where $\ket{\psi_{\text{max}}}$ is an element of a mutually unbiased basis for the energy eigenbasis. We can see that coherence in $\omega_R$ is always below the $2.5\%$ of the maximum possible coherence, and decreases as a function of the disorder, as expected. {\bf Right.} We additionally compute the measure appearing in Prop.~\ref{prop:csl_semi-optimality}, that is, the trace distance between $\omega_R$ and its decohered version $\Delta(\omega_R)$. We find that, for high disorder strengths (that is, during the MBL phase) this measure is small and it stays almost constant as $|R|$ increases (with the exception of the increase occurring between $|R| = 1$ and $2$). Both measures have been derived for $100$ random realization of the landscape potential.}
	\label{fig:coh_reduced}
\end{figure}
\par
The other assumption in Prop.~\ref{prop:csl_optimality} and \ref{prop:csl_semi-optimality} involves the energy gaps of the region Hamiltonian $H_R$. Due to the noise affecting the Hamiltonian of the spin chain, it seems a reasonable assumption to have non-degenerate energy gaps which could, at least up to a given $n \in \N$, satisfy this condition. A more detail discussion on this property is given in the main text, and in the following we clarify the limitations of the swapping collision model studied in this section.
\begin{remark}[Role of trivial Hamiltonians]\label{rem:trivH}
	For the case of trivial Hamiltonians, which clearly violates the ESC, the target thermal state is simply
	the maximally mixed state. Since we allow for the use of random unitary channels, thermalization in this case can already
	be achieved without further bath copies. Another way to put it is that the usage of bath copies and randomness coincide
	fully. The results in Ref.~\cite{boes_catalytic_2018} imply that only an amount of $ \log d $ of randomness (bits), are needed
	to perform this task, where $d = \dim \mcH$.
\end{remark}
\subsection{Limitation of the stochastic swapping collision model}
We have seen that, while the channel $\bar{\chn}_n$ can be proven to be optimal, necessary conditions on the Hamiltonian
follow. When such conditions are dropped, we in fact know of cases where this channel is non-optimal (see
Remark~\ref{rem:trivH} for example), in the sense that there exists a much more efficient protocol using a
smaller number of bath copies. The case in Remark~\ref{rem:trivH} is somewhat less interesting since considering
trivial Hamiltonians reduces the bath to a purely randomness resource, without any heat considerations whatsoever.
In this section, we provide a counter example using non-trivial Hamiltonians.
\par
One implication of the restriction given by the ESC of Def.~\ref{def:ESC} is that energy gaps of a single copy of the
system cannot be degenerate. The following counter-example is constructed then as follows; let $ H^{(1)} =
\sum_{i=1}^4 E_i \Pi_{E_i}$ be a 4-level Hamiltonian such that $ E_2 - E_1 = E_4 - E_3 $. Furthermore, let $ \tau $
be a particular thermal state of $ H^{(1)} $ with eigenvalues $ \left( p_1, \cdots, p_4 \right)$ denoting the thermal
occupations. To show that $\bar{\chn}_n$ can be non-optimal in general, let us consider simply the usage of one
copy of the bath.  Due to the degeneracy of energy gaps assumed, the energy subspace for global energy 
\begin{equation}
\label{eq:deg_subspace}
E_t = E_1 + E_4 = E_2+E_3 ,
\end{equation}
contains two different types (as characterized by the tuple in Eq.~\eqref{eq:def_tuple_type}). In particular, if we
denote $ \Pi_{i,j} = |i,j\rangle\langle i,j| + |j,i\rangle\langle j,i|$, then the projector on energy subspace $ E_t $ can
be written as $ \Pi_{E_t} = \Pi_{1,4} + \Pi_{2,3} $.
\par
Let $ \tau^{(2)} := \tau^{\otimes 2} $ denote the target state. We know that the total weight within
this subspace is $ p_1p_4+p_2p_3+p_3p_2+p_4p_1 = 4p_1p_4 $, where we know that since $ \tau $ is
a thermal state, Eq.~\eqref{eq:deg_subspace} implies that $ p_2p_3 = p_1p_4 $. Therefore, within the
$ E_t $ energy subspace, we have
\begin{equation}
\Pi_{E_t} \, \tau^{(2)} \, \Pi_{E_t} = p_1 p_4 \, \Pi_{E_t},
\end{equation}
which is a uniform distribution. What about the output state of the channel $\bar{\chn}_2$? Assuming
an initial state $ \rho = \sum_i q_i |i\rangle\langle i| $, the stochastic swapping process produces an
output state $ \rho^{(2)} = \bar{\chn}_2(\rho) = \frac{1}{2} \left( \rho\otimes\tau + \tau\otimes\rho \right)$.
Within the same energy subspace, the total weight is $ q_1p_4+q_2p_3+q_3p_2+q_4p_1 $, and the
distribution is
\begin{equation}
\Pi_{E_t} \, \rho^{(2)} \, \Pi_{E_t} =  \frac{q_1 p_4+q_4p_1}{2} \cdot \Pi_{1,4} + \frac{q_2p_3+q_3p_2}{2}\Pi_{2,3},
\end{equation}
which is uniform across the subspace for each type.
\begin{figure}[t]
	\includegraphics[width=0.5\textwidth]{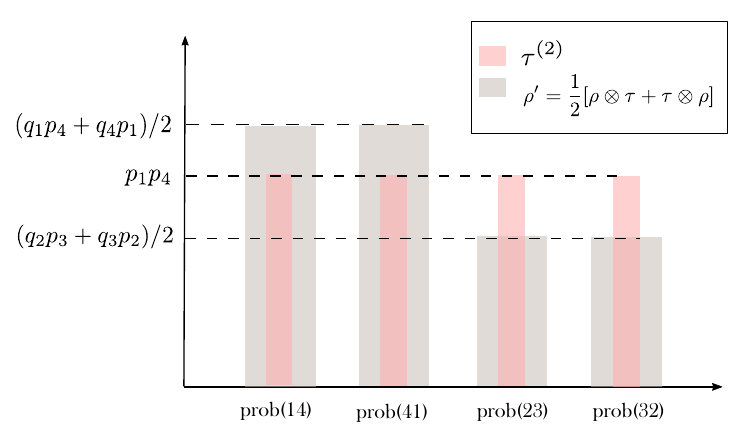}
	\caption{Comparison of eigenvalues of $ \tau^{(2)} $ and $ \rho^{(2)} $ within the energy $ E_t $ subspace. If each individual eigenvalues of $ \rho^{(2)} $ are larger (or smaller) than $ \tau^{(2)} $, then smoothening the subspace on $ \rho^{(2)} $ does not affect the trace distance. However, if we have the situation in the diagram where some eigenvalues of $ \rho^{(2)}$ are larger than that of $ \tau^{(2)} $, while some are smaller, then one can further reduce the trace distance by smoothening out the subspace on $ \rho^{(2)} $.}
	\label{fig:counter_ex}
\end{figure}
\par
Supplementary Figure~\ref{fig:counter_ex} shows the eigenvalues in this 4-dimensional subspace, compared to the distribution
given by $ \tau^{(2)} $. Note that as long as
\begin{equation}
\frac{q_1p_4+q_4p_1}{2} > p_1p_4 > \frac{q_2p_3+q_3p_2}{2}
\quad \text{or} \quad
\frac{q_1p_4+q_4p_1}{2} < p_1p_4 < \frac{q_2p_3+q_3p_2}{2},
\end{equation}
holds, one can always further decrease the trace distance by using another state $ \rho^* $ that makes the
entire subspace uniform. An example of this is when $ q_1 =1,q_2=q_3=q_4=0 $, while assuming that $ p_1\leq \frac{1}{2} $,
so that $ p_1 p_4 \leq \frac{p_4}{2} $. Note that since $ p_1 $ is the largest eigenvalue of the single-copy thermal
state, $ p_1 \geq \frac{1}{4} $ as well. Concretely, take the state $ \rho^* = \rho^{(2)} - \Pi_{E_t} \, \rho^{(2)} \, \Pi_{E_t}
+ \frac{p_4}{4}\Pi_{E_t} $. Then 
\begin{align*}
d(\rho^*,\tau^{(2)}) &= d(\rho^{(2)},\tau^{(2)}) - \frac{1}{2}\tr(|\Pi_{E_t} (\rho^{(2)}-\tau^{(2)}) \Pi_{E_t}|)
+ \frac{1}{2}\tr(|\frac{p_4}{4}\Pi_{E_t} -\Pi_{E_t}\tau^{(2)}\Pi_{E_t}|) \\
&=d(\rho^{(2)},\tau^{(2)})-\frac{p_4}{2}+ p_4(2p_1 - \frac{1}{2}) 
=d(\rho,\tau^{(2)})+ p_4 (2p_1 -1) \leq d(\rho^{(2)},\tau^{(2)}),
\end{align*}
which implies that the output of the stochastic swapping channel $ \rho^{(2)} $ cannot be optimal in terms of
minimizing the distance.
\end{document}